\documentclass[12pt]{article}
\usepackage{amsmath}
\usepackage{graphicx}%
\usepackage{amsfonts}%
\usepackage{amssymb}
\usepackage{overpic}
\usepackage{subfigure}
\usepackage{rotating}
\usepackage{url}
\usepackage{fullpage}
\usepackage{comment}
\usepackage{enumerate}
\usepackage{multirow}
\usepackage{booktabs}
\usepackage{tabularx}
\usepackage{makecell}
\usepackage{booktabs}
\usepackage{natbib}
\usepackage[colorlinks = true,
            linkcolor = blue,
            urlcolor  = blue,
            citecolor = blue,
            anchorcolor = blue]{hyperref}
\usepackage[letterpaper,left=1in,top=1in,right=1in,bottom=1in,lines=25]{geometry}

\usepackage{setspace}
\newcommand{\bA}{ {\boldsymbol A} }

\newcommand{\bB}{ {\boldsymbol B} }

\newcommand{\bC}{ {\boldsymbol C} }

\newcommand{\bD}{ {\boldsymbol D} }

\newcommand{\bH}{ {\boldsymbol H} }

\newcommand{\bI}{ {\boldsymbol I} }

\newcommand{\bm}{ {\boldsymbol m} }
\newcommand{\bM}{ {\boldsymbol M} }

\newcommand{\bS}{ {\boldsymbol S} }

\newcommand{\bu}{ {\boldsymbol u} }
\newcommand{\bU}{ {\boldsymbol U} }

\newcommand{\bw}{ {\boldsymbol w} }
\newcommand{\bW}{ {\boldsymbol W} }
\newcommand{\bx}{ {\boldsymbol x} }
\newcommand{\bX}{ {\boldsymbol X} }
\newcommand{\by}{ {\boldsymbol y} }

\newcommand{\bz}{ {\boldsymbol z} }

\newcommand{\bgamma}{ {\boldsymbol \gamma} }
\newcommand{\bGamma}{ {\boldsymbol \Gamma} }
\newcommand{\bdelta}{ {\boldsymbol \delta} }

\newcommand{\bLambda}{ {\boldsymbol \Lambda} }
\newcommand{\bmu}{ {\boldsymbol \mu} }

\newcommand{\bSigma}{ {\boldsymbol \Sigma} }

\newcommand{\bzero}{ {\boldsymbol 0} }

\newcommand{\given}{\,|\,}

\newtheorem{theorem}{Theorem}[section]
\newtheorem{lemma}[theorem]{Lemma}

\newenvironment{proof}[1][Proof]{\begin{trivlist}
\item[\hskip \labelsep {\bfseries #1}]}{\end{trivlist}}

\newcommand{\qed}{\nobreak \ifvmode \relax \else
      \ifdim\lastskip<1.5em \hskip-\lastskip
      \hskip1.5em plus0em minus0.5em \fi \nobreak
      \vrule height0.75em width0.5em depth0.25em\fi}

\title{Bayesian Regression with Undirected Network Predictors with an Application to Brain Connectome Data}
\author{Sharmistha Guha and Abel Rodriguez}
\date{} 
\begin{document}
\maketitle





\begin{abstract}
This article proposes a Bayesian approach to regression with a continuous scalar response and an undirected network predictor. Undirected network predictors are often expressed in terms of symmetric adjacency matrices, with rows and columns of the matrix representing the nodes, and zero entries signifying no association between two corresponding nodes. Network predictor matrices are typically vectorized prior to any analysis, thus failing to account for the important structural information in the network. 
This results in poor inferential and predictive performance in presence of small sample sizes.
We propose a novel class of \emph{network shrinkage priors} for the coefficient corresponding to the undirected network predictor. The proposed framework is devised to detect both nodes and edges in the network predictive of the response. Our framework is implemented using an efficient Markov Chain Monte Carlo algorithm. Empirical results in simulation studies illustrate strikingly superior inferential and predictive gains of the proposed framework in comparison with the ordinary high dimensional Bayesian shrinkage priors and penalized optimization schemes. We apply our method to a brain connectome dataset that contains information on brain networks along with a measure of creativity for multiple individuals. Here, interest lies in building a regression model of the creativity measure on the network predictor to identify important
regions and connections in the brain strongly associated with creativity. To the best of our knowledge, our approach is the first principled Bayesian method that is able to detect scientifically interpretable regions and connections in the brain actively impacting the continuous response (creativity) in the presence of a small sample size.
\end{abstract}

\noindent \emph{Keywords:} Brain connectome; Edge selection; High dimensional regression; Network predictors; Network shrinkage prior; Node selection.

\section{Introduction}\label{intro}
In recent years, network data has become ubiquitous in disciplines as diverse as neuroscience, genetics, finance and economics. Nonetheless, statistical models that involve network data are particularly challenging, not only because they require dimensionality reduction procedures to effectively deal with the large number of pairwise relationships, but also because flexible formulations are needed to account for the topological structure of the network.


The literature has paid heavy attention to models that aim to understand the relationship between node-level covariates and the structure of the network.  A number of classic models treat the dyadic observations as the response variable, examples include random graph models \citep{erdos1960evolution}, exponential random graph models \citep{frank1986markov}, social space models \citep{hoff2002latent,hoff2005bilinear,hoff2009hierarchical} and stochastic block models \citep{nowicki2001estimation}.  The goal of these models is often either to predict unobserved links or to investigate  \textit{homophily}, i.e., the process of formation of social ties due to matching individual traits.  Alternatively, models that investigate \textit{influence} or \textit{contagion} attempt to explain the node-specific covariates as a function of the network structure (e.g., see \citealp{christakis2007spread}; \citealp{fowler2008dynamic}; \citealp{shoham2015modeling} and references therein).  Common methodological approaches in this context include simultaneous autoregressive (SAR) models (e.g., see \citealp{lin2010identifying}) and threshold models (e.g., see \citealp{watts2009threshold}).  However, ascertaining the direction of a causal relationship between network structure and link or nodal attributes, i.e., whether it pertains to homophily or contagion, is difficult (e.g., see \citealp{doreian2001causality} and \citealp{shalizi2011homophily} and references therein).  Hence, there has been a growing interest in joint models for the coevolution of the network structure and nodal attributes (e.g., see \citealp{fosdick2015testing,durante2017bayesian,de2010obesity,niezink2016,guhanrod2017}).

In this paper we investigate Bayesian models for network regression.  Unlike the problems discussed above, in network regression we are interested in the relationship between the structure of the network and one or more global attributes of the experimental unit on which the network data is collected.  As a motivating example, we consider the problem of predicting the composite creativity index of individuals on the basis of neuroimaging data reassuring the connectivity of different brain regions.  The goal of these studies is twofold.  First, neuroscientists are interested in identifying regions of the brain that are involved in creative thinking.  Secondly, it is important to determine how the strength of connection among these influential regions affects the level of creativity of the individual.  More specifically, we construct a novel Bayesian \textit{network shrinkage prior} that combines ideas from spectral decomposition methods and spike-and-slab priors to generate a model that respects the structure of the predictors.  The model produces accurate predictions, allows us to identify both nodes and links that have influence on the response, and yield well-calibrated interval estimates for the model parameters.

A common approach to network regression is to use a few summary measures from the network in the context of a flexible regression or classification approach (see, for example, \citealp{bullmore2009complex} and references therein).  Clearly, the success of this approach is highly dependent on selecting the right summaries to include.  Furthermore, this kind of approach cannot identify the impact of specific nodes on the response, which is of clear interest in our setting.  Alternatively, a number of authors have proceeded to vectorize the network predictor (originally obtained in the form of a symmetric matrix). Subsequently, the continuous response would be regressed on the high dimensional collection of edge weights (e.g., see \citealp{richiardi2011decoding} and \citealp{craddock2009disease}).  This approach can take advantage of the recent developments in high dimensional regression, consisting of both penalized optimization \citep{tibshirani1996regression} and Bayesian shrinkage \citep{park2008bayesian,carvalho2010horseshoe,armagan2013generalized}.  However, this approach treats the links of the network as if they were exchangeable, ignoring the fact that coefficients that involve common nodes can be expected to be correlated a priori. Ignoring this correlation often leads to poor predictive performance and can potentially impact model selection.

Recently, \citealp{relion2017network} proposed a penalized optimization scheme that not only enables classification of networks, but also identifies important nodes and edges. Although this model seems to perform well for prediction problems, uncertainty quantification is difficult because standard bootstrap methods are not consistent for Lasso-type methods (e.g., see \citealp{kyung2010penalized} and \citealp{chatterjee2010asymptotic}). Modifications of the bootstrap that produce well-calibrated confidence intervals in the context of standard Lasso regression have been proposed (e.g., see \citealp{chatterjee2011bootstrapping}), but it is not clear whether they extend to the kind of group Lasso penalties discussed in \citealp{relion2017network}. Recent developments on tensor regression (e.g., see \citealp{zhou2013tensor,guhaniyogi2017bayesian}) are also relevant to our work. However, these approaches tend to focus mainly on prediction and identification of important edges, but are not designed to detect important nodes impacting the response.

The rest of the article evolves as follows. 
Section \ref{sec3} proposes the novel network shrinkage prior and discusses posterior computation for the proposed model. Empirical investigations with various simulation studies are presented in Section \ref{sec4}, while Section \ref{sec5} analyzes the brain connectome dataset. We provide results on \emph{region of interest} (ROI) and \emph{edge} selection and find them to be scientifically consistent with previous studies. Finally, Section \ref{sec6} concludes the article with an eye towards future work.

\section{Model Formulation}\label{sec3}
Let $y_i\in\mathbb{R}$ and $\bA_i$ represent the observed continuous scalar response and the corresponding weighted undirected network for the $i$th sample, $i=1,...,n$ respectively. All graphs share the same labels on their nodes. For example, in our brain connectome application discussed subsequently, $y_i$ corresponds to a phenotype, while $\bA_i$ encodes the network connections between different regions of the brain for the $i$th individual. Let the network corresponding to any individual consist of $V$ nodes. Mathematically, this amounts to $\bA_i$ being a $V \times V$ matrix, with the $(k,l)$th entry of $\bA_i$ denoted by $a_{i,k,l}\in\mathbb{R}$. We focus on networks that contain no self relationship, i.e., $a_{i,k,k} \equiv 0$, and are undirected ($a_{i,k,l}=a_{i,l,k}$). The brain connectome application considered here naturally justifies these assumptions. Although we present our model specific to these settings, it will be evident that the proposed model can be easily extended to directed networks with self-relations.

\subsection{Bayesian Network Regression Model}\label{sec222}

We propose the high dimensional regression model of the response $y_i$ for the $i$-th individual on the undirected network predictor $\bA_i$ as
\begin{align}\label{initial_model}
y_i=\mu+\langle\bA_i, \bB\rangle_F+\epsilon_i,\:\:\epsilon_i\stackrel{iid}{\sim} N(0,\tau^2),
\end{align}
where $\bB$ is the network coefficient matrix of dimension $V\times V$ whose $(k,l)$th element is given by $\beta_{k,l}$ and $\langle\bA_i, \bB\rangle_F = Trace(\bB'\bA_i)$ denotes the Frobenius inner product between $A_i$ and $B$. The Frobenius inner product is the natural inner product in the space of matrices and is a generalization of the dot product from vector to matrix spaces. 
Similar to the network predictor, the network coefficient matrix $\bB$ is assumed to be symmetric with zero diagonal entries. The parameter $\tau^2$ is the variance of the observational error.
Since self relationship is absent and both $\bA_i$ and $\bB$ are symmetric, $\langle\bA_i, \bB\rangle_F=2\sum\limits_{1 \leq k < l \leq V} a_{i,k,l} \beta_{k,l}$. Then, denoting $\gamma_{k,l}=2\beta_{k,l}$, (\ref{initial_model}) can be rewritten as
\begin{align}\label{model1}
y_i = \mu+\sum\limits_{1 \leq k < l \leq V} a_{i,k,l} \gamma_{k,l} + \epsilon_i,\: \epsilon_i \sim N(0, \tau^2),
\end{align}
Equation (\ref{model1}) connects the network regression model with the linear regression framework with $a_{i,k,l}$'s as predictors and $\gamma_{k,l}$'s as the corresponding coefficients. However, while in ordinary linear regression the predictor coefficients are indexed by the natural numbers $\mathcal{N}$, Model (\ref{model1}) indexes the predictor coefficients by their positions in the matrix $\bB$. This is done in order to keep a tab not only on the edge itself but also on the nodes connecting the edges. 

\subsection{Developing the Network Shrinkage Prior}\label{sec223}
\subsubsection{Vector Shrinkage Prior}
High dimensional regression with vector predictors has recently been of interest in Bayesian statistics.
An overwhelming literature in Bayesian statistics in the last decade has focused on shrinkage priors which shrink
coefficients corresponding to unimportant variables to zero while minimizing the shrinkage of coefficients corresponding to influential variables. Many of these shrinkage
prior distributions can be expressed as a scale mixture of normal distributions, commonly referred to as
\emph{global-local scale mixtures} (\citealp{polson2010shrink}), that enable fast computation employing
simple conjugate Gibbs sampling. More precisely, in the context of model (\ref{model1}), a global-local scale mixture prior
would take the form
\begin{align*}
\gamma_{k,l} \sim N(0,s_{k,l}\tau^2),\:\:\:\:s_{k,l} \sim g_1,\:\:\:\:\tau^2 \sim g_2 ,\:\:\:\:1 \leq k < l \leq V,
\end{align*}
Note that $s_{1,2},...,s_{V-1,V}$ are local scale parameters controlling the shrinkage of the coefficients, while $\tau^2$ is the global scale parameter.
Different choices of $g_1$ and $g_2$ lead to different classes of Bayesian shrinkage priors which have appeared in the literature. For example, the Bayesian Lasso (\citealp{park2008bayesian}) prior takes $g_1$ as exponential and $g_2$ as the Jeffreys prior, the Horseshoe prior (\citealp{carvalho2010horseshoe}) takes both $g_1$ and $g_2$ as half-Cauchy distributions and the Generalized Double Pareto Shrinkage prior (\citealp{armagan2013generalized}) takes $g_1$ as exponential and $g_2$ as the Gamma distribution.

The direct application of this global-local prior in the context of (\ref{model1}) is unappealing. To elucidate further, note that if node $k$ contributes minimally to the response, one would expect to
have smaller estimates for all coefficients $\gamma_{k,l}$, $l>k$ and  $\gamma_{l',k}$, $l'< k$ corresponding to edges connected to node $k$. The ordinary global-local shrinkage prior distribution given as above does not necessarily conform to such an important restriction.
In what follows, we build a network shrinkage prior upon the existing literature that respects this constraint as elucidated in the next section.

\subsubsection{Network Shrinkage Prior}\label{network_shrinkage}
We propose a shrinkage prior on the coefficients $\gamma_{k,l}$ and refer to it as the \emph{Bayesian Network Shrinkage prior}. The prior borrows ideas from low-order spectral representations of matrices.
Let $\bu_1,...,\bu_V\in\mathbb{R}^R$ be a collection of $R$-dimensional latent variables, one for each node, such that $\bu_k$ corresponds to node $k$. We draw each $\gamma_{k,l}$ conditionally independent from a density that can be represented as a location and scale mixture of normals. More precisely,
\begin{align}\label{first-stage}
\gamma_{k,l}|s_{k,l},\bu_k,\bu_l,\tau^2\sim N(\bu_k' \bLambda \bu_l,\tau^2 s_{k,l}),
\end{align}
where $s_{k,l}$ is the scale parameter corresponding to each $\gamma_{k,l}$ and $\bLambda=\mbox{diag}(\lambda_1,...,\lambda_R)$ is an $R\times R$ diagonal matrix. The diagonal elements $\lambda_r\in\{0,1\}$'s are introduced to assess the effect of the $r$th dimension of $\bu_k$ on the mean of $\gamma_{k,l}$. In particular, $\lambda_r=0$ implies that the $r$th dimension of the latent variable $\bu_k$ is not informative for any $k$. Note that if $s_{k,l}=0$, this prior will imply that $\bGamma=2\bB=\bU'\bLambda\bU$, where $\bU$ is an $R\times V$ matrix whose $k$th column corresponds to $\bu_k$ and $\bGamma=((\gamma_{k,l}))_{k,l=1}^{V}$. Since $R \ll V$, the mean structure of $\bGamma$ assumes a low-rank matrix decomposition.

In order to learn which components of $\bu_k$ are informative for (\ref{first-stage}), we assign a hierarchical prior
\begin{align*}
\lambda_r \sim Ber (\pi_{r}),\:\:\:\:\pi_{r} \sim Beta(1, r^{\eta}),\:\:\:\:\eta>1.
\end{align*}

The choice of hyper-parameters of the beta distribution is crucial in order to impart increasing shrinkage on $\lambda_r$ as $r$ grows.
In particular, note that $E[\lambda_r]=1/(1+r^{\eta})\rightarrow 0$, as $r\rightarrow\infty$, so that the prior favors choice of smaller number of active components in $\bu_k$'s impacting the response. Note that $\sum_{r=1}^{R}\lambda_r$ is the number of dimensions of $\bu_k$ contributing to predict the response. We refer to $\sum_{r=1}^{R}\lambda_r$ as $R_{eff}$, the \emph{effective dimensionality} of the latent variables. The choice of prior hyperparameters of $\lambda_r$ ensures that $var(\gamma_{k,l})$'s remain finite even as $R\rightarrow\infty$. In fact, if one assumes $var(\bu_k)<\infty$, i.e., a proper prior is assigned on $\bu_k$ (which will hold in our case as will be evident later), some routine algebra yields that $var(\gamma_{k,l})$ is finite when $R\rightarrow\infty$ if $\sum_{r=1}^{R}var(\lambda_r)<\infty$ as $R\rightarrow\infty$. In fact, the hyperparameters of the beta prior on $\pi_r$ are such that
$\sum_{r=1}^{R}var(\lambda_r)=\sum_{r=1}^{R}\frac{r^{\eta}}{(1+r^{\eta})^2(2+r^{\eta})}<\infty$ as $R\rightarrow\infty$.


The mean structure of $\gamma_{k,l}$ is constructed to take into account the interaction between the $k$th and the $l$th nodes. 
Drawing intuition from \citealp{hoff2005bilinear}, we might imagine that the interaction between the $k$th and $l$th nodes has a positive, negative or neutral impact on the response depending on whether $\bu_k$ and $\bu_l$ are in the same direction, opposite direction or orthogonal to each other respectively. In other words, whether the angle between
$\bu_k$ and $\bu_l$ is acute, obtuse or right, i.e., $\bu_k'\bLambda\bu_l>0$, $\bu_k'\bLambda\bu_l<0$ or $\bu_k'\bLambda\bu_l=0$ respectively. The conditional mean
of $\gamma_{k,l}$ in (\ref{first-stage}) is constructed to capture such latent network information in the predictor.

In order to select active nodes (i.e., to determine if a node is inactive in explaining the response), we assign the \emph{spike-and-slab}  (\citealp{ishwaran2005spike})  mixture prior on the latent factor $\bu_k$ as below
\begin{align}
\bu_k\sim\left\{\begin{array}{cc}
N(\bzero, \bM), & \mbox{if}\:\: \xi_k=1\\
 \delta_{\bzero}, & \mbox{if}\:\: \xi_k=0\\
 \end{array}
 \right.,\:\:\:\: \xi_k\sim Ber(\Delta),
\end{align}
where $\delta_{\bzero}$ is the Dirac-delta function at $\bzero$ and $\bM$ is a covariance matrix of order $R\times R$. The parameter $\Delta$ corresponds to the probability of the nonzero mixture component. Note that if the $k$th node of the network predictor is inactive in predicting the response, then a-posteriori $\xi_k$ should provide high probability to $0$. Thus, based on the posterior probability of $\xi_k$, it will be possible to identify unimportant nodes in the network regression.
The rest of the hierarchy is accomplished by assigning prior distributions on the
$s_{k,l}$'s and $\bM$ as follows:  
\begin{align*}
s_{k,l} \sim Exp(\theta^2/2),\:\:
 \bM \sim IW(\bS,\nu),\:\:
\theta^2 \sim Gamma(\zeta,\iota),\:\:
 \Delta \sim Beta(a_{\Delta},b_{\Delta}),
\end{align*}
where $\bS$ is an $R\times R$ positive definite scale matrix. $IW(\bS,\nu)$ denotes an Inverse-Wishart distribution with scale matrix $\bS$ and degrees of freedom $\nu$.  Finally, we choose a non-informative prior on $(\mu, \tau^2)$ such that $p(\mu, \tau^2) \propto \frac{1}{\tau^2}$. Appendix B shows the propriety of the posterior distribution under this prior.

Note that, if $\bu_k'\bLambda\bu_l=0$, the marginal prior distribution of $\gamma_{k,l}$ integrating out all the latent variables turns out to be the double exponential distribution which is connected to the Bayesian Lasso prior. When 
$\bu_k'\bLambda\bu_l \neq 0$, the marginal prior distribution of $\gamma_{k,l}$ appears to be a location mixture of double exponential prior distributions, the mixing distribution being a function of the network. Owing to this fact, we coin the proposed prior distribution as the \emph{Bayesian Network Lasso} prior.

\subsection{Posterior Computation}
Although summaries of the posterior distribution cannot be computed in closed form, full conditional distributions for all parameters are available and correspond to standard families. Thus the posterior computation of parameters proceeds through Gibbs sampling. Details of all the full conditionals are presented in Appendix A.

In order to identify whether the $k$th node is important in terms of predicting the response, we rely on the post burn-in $L$ samples $\xi_k^{(1)},....,\xi_k^{(L)}$ of $\xi_k$. Node $k$ is recognized to be influential if $\frac{1}{L}\sum_{l=1}^{L}\xi_k^{(l)}>0.5$.
On the other hand, an estimate of $P(R_{eff}=r\given Data)$ is given by $\frac{1}{L}\sum_{l=1}^{L}I(\sum_{m=1}^{R}\lambda_m^{(l)}=r)$, where $I(A)$ for an event $A$ is 1 if the event $A$ happens and $0$ otherwise and $\lambda_m^{(1)},...,\lambda_m^{(L)}$ are the $L$ post burn-in MCMC samples of $\lambda_m$.

\section{Simulation Studies}\label{sec4}
This section comprehensively contrasts both the inferential and the predictive performances of our proposed approach with a number of competitors in various simulation settings. We refer to our proposed approach as the \emph{Bayesian Network Regression} (BNR). As competitors, we consider both penalized likelihood methods as well as Bayesian shrinkage priors for high-dimensional regression.

Our first set of competitors use generic variable selection and shrinkage methods that treat edges between nodes as ``bags of predictors" to run high dimensional regression, thereby ignoring the relational nature of the predictor. More specifically, we use Lasso (\citealp{tibshirani1996regression}), which is a popular penalized optimization scheme, and the Bayesian Lasso (\citealp{park2008bayesian}) and Horseshoe priors (\citealp{carvalho2010horseshoe}), which are popular Bayesian shrinkage regression methods. In particular, the Horseshoe is considered to be the state-of-the-art Bayesian shrinkage prior and is known to perform well, both in sparse and not-so-sparse regression settings. We use the \texttt{glmnet} package in \texttt{R} (\citealp{friedman2010regularization}) to implement Lasso regression, and the \texttt{monomvn} package  in \texttt{R} (\citealp{gramacy2013package}) to implement the Bayesian Lasso (BLasso for short) and the Horseshoe (BHS for short).
A thorough comparison with these methods will indicate the relative advantage of exploiting the structure of the network predictor.

Additionally, we compare our method to a frequentist approach that develops network regression in the presence of a network predictor and scalar response (\citealp{relion2017network}). To be precise, we adapt \citealp{relion2017network} to a \emph{continuous response} context and propose to estimate the network regression coefficient matrix $\bB$ by solving
\begin{align}\label{relion}
&\hat{\bB}=\arg\min\limits_{\bB\in\mathbb{R},\bB=\bB',diag(\bB)={\boldsymbol 0}}\Bigg\{\frac{1}{n}\sum_{i=1}^{n}(y_i-\mu-\langle \bA_i,\bB\rangle_F)^2+\nonumber\\
&\qquad\qquad\qquad\qquad\qquad\qquad\frac{\varphi}{2}||\bB||_F^2+\varsigma\left(\sum_{k=1}^V||\bB_{(k)}||_2+\rho||\bB||_1\right)\Bigg\},
\end{align}
where $||\bB||_F=\sqrt{\langle\bB, \bB\rangle_F}$ denotes the Frobenius norm, $||\bB||_1$ is the sum of the absolute values of all the elements of matrix $\bB$, $||\cdot||_2$ is the $l_2$ norm of a vector, $\bB_{(k)}$ is the $k$th row of $\bB$ and $\varphi,\rho,\varsigma$ are tuning parameters. The best possible choice of the tuning parameter triplet $(\varphi,\rho,\varsigma)$ is made using cross validation over a grid of possible values.  \citealp{relion2017network} argue that the penalty in (\ref{relion}) incorporates the network information of the predictor, thereby yielding superior inference to any ordinary penalized optimization scheme. Hence comparison with (\ref{relion}) will potentially highlight the advantages of a carefully structured Bayesian network shrinkage prior over the penalized optimization scheme incorporating network information. 
In the absence of any open source code, we implement the algorithm in \citealp{relion2017network}. 
All Bayesian competitors are allowed to draw $50,000$ MCMC samples out of which the first $30,000$ are discarded as burn-ins. Convergence is assessed by comparing different simulated sequences of representative parameters started at different initial values (\citealp{gelman2014bayesian}). All posterior inference is carried out based on the rest $20,000$ MCMC samples after suitably thinning the post burn-in chain. We monitor the auto-correlation plots and effective sample sizes of the iterates, and they are found to be satisfactorily uncorrelated.

\subsection{Predictor and Response Data Generation}\label{secgen}
In all simulation studies, the undirected symmetric network predictor $\bA_i$ for the $i$th sample is simulated by drawing $a_{i,k,l}\stackrel{iid}{\sim} N(0,1)$ for
$k<l$ and setting $a_{i,k,l}=a_{i,l,k}$ and $a_{i,k,k}=0$ for all $k,l\in\left\{1,...,V\right\}$. The response $y_i$ is generated according to the network regression model
\begin{align}\label{data_gen}
y_i = \mu_0 + \langle \bA_i,\bB_0\rangle_F + \epsilon_i;\: \epsilon_i \sim N(0, \tau_0^2),
\end{align}
with $\tau_0^2$ as the true noise variance. In a similar vein as section \ref{sec222}, define $\gamma_{k,l,0} = 2 \beta_{k,l,0}$, which will later be used to define the mean squared error. In all of our simulations, we use $V=20$ nodes and $n=70$ samples.

To study all competitors under various data generation schemes, we simulate the true network predictor coefficient $\bB_0$ under three simulation scenarios, referred to as
\emph{Simulation 1}, \emph{Simulation 2} and \emph{Simulation 3}.\\

\noindent \underline{\emph{Simulation 1}}\\
In \emph{Simulation 1}, we draw $V$ latent variables $\bw_k$, each of dimension $R_{gen}$, 
from a mixture distribution given by
\begin{align}\label{wk}
\bw_k \sim \pi_w N_{R_{gen}} (\bw_{mean}, \bw_{sd}^2) + (1 - \pi_w) \bdelta_{{\boldsymbol 0}}; \: k \in \{1,...,V\},
\end{align}
where $\bdelta_{{\boldsymbol 0}}$ is the Dirac-delta function and $\pi_w$ is the probability of any $\bw_k$ being nonzero. The elements of the network predictor coefficient
$\bB_0$ are given by $\beta_{k,l,0}= \beta_{l,k,0} = \frac{\bw_k'  \bw_l}{2}$. 
Note that if $\bw_k$ is zero, any edge connecting the $k$th node has no contribution to the regression mean function in (\ref{data_gen}), i.e., the $k$th node becomes inactive in predicting the response. Accordingly, $(1-\pi_w)$ is the probability of a node being inactive. Hence, hereafter, $(1-\pi_w)$ is referred to as the \emph{sparsity} parameter in the context of the data generation mechanism under \emph{Simulation 1}.

It is worth mentioning that in \emph{Simulation 1}, the simulated coefficient $\beta_{k,l,0}$ corresponding to the $(k,l)$th edge $a_{i,k,l}$ represents a bi-linear interaction between the latent variables corresponding to the $k$th and the $l$th nodes. Thus \emph{Simulation 1} generates $\bB_0$ respecting the network structure in $\bA_i$. 
For a comprehensive picture of \emph{Simulation 1}, we consider $9$ different cases as summarized in Table~\ref{Tab1}. In each of these cases, the network predictor and the response are generated by changing the sparsity $\pi_w$ and the true dimension $R_{gen}$ of the latent variables $\bw_k$'s. The table also presents the maximum dimension $R$ of the latent variables $\bu_k$ for the fitted network regression model (\ref{model1}). Note that the various cases also allow model mis-specification with unequal choices of $R$ and $R_{gen}$. For all simulations, $\bw_{mean}$ and $\bw_{sd}^2$ are set as $0.8 \times \mathbf{1}_{Rgen}$ and $\bI_{R_{gen}\times R_{gen}}$, respectively. Apart from investigating the model as a tool for inference and prediction, it is of interest to observe the posterior distributions of $\lambda_r$'s to judge if the model effectively learns the dimensionality $R_{gen}$ of the latent variables $\bw_k$. Noise variance $\tau_0^2$ is fixed at $1$ for all scenarios. 
\\

\noindent \underline{\emph{Simulation 2}}\\
In \emph{Simulation 2}, we draw $V$ indicator variables $\xi_1^0,...,\xi_V^0\stackrel{iid}{\sim} Ber(\pi_{w^*})$ corresponding to $V$ nodes of the network. If both $\xi_k^0=1$ and $\xi_l^0=1$, the edge coefficient $\beta_{k,l,0}$ connecting the $k$th and the $l$th nodes ($k<l$) is simulated from $N(0.8,1)$. Respecting the symmetry condition, we set $\beta_{l,k,0}=\beta_{k,l,0}$ for nonzero edge coefficients. If $\xi_k^0=0$ for any $k$, we set $\beta_{k,l,0}=\beta_{l,k,0}=0$ for any $l$.  The edge between any two nodes $k$ and $l$ does not contribute in predicting the response if either $\xi_k^0$ or $\xi_l^0$ is zero. Hence in the context of \emph{Simulation 2}, $1-\pi_{w^*}$ is referred to as the \emph{sparsity} parameter.

Under \emph{Simulation 2}, the network structure of the nonzero coefficients is lost as there is no impact on the regression function in (\ref{data_gen}) from any edge that is connected to any inactive node. The coefficient corresponding to the $(k,l)$th edge $a_{i,k,l}$ is drawn from a normal distribution with no interaction specific to nodes $k$ and $l$. Thus, the data generation scheme will presumably offer no advantage to our model. We present two cases of \emph{Simulation 2} in Table~\ref{Tabnew2}. Noise variance $\tau_0^2$ is fixed at $1$ for all scenarios. \\

\noindent \underline{\emph{Simulation 3}}\\
In \emph{Simulation 3}, we draw $V$ indicator variables $\xi_1^0,...,\xi_V^0\stackrel{iid}{\sim} Ber(\pi_{w^{**}})$ corresponding to $V$ nodes of the network. If both $\xi_k^0=1$ and $\xi_l^0=1$, then the edge coefficient $\beta_{k,l,0}$ connecting the $k$th and the $l$th nodes ($k<l$) is simulated from a mixture distribution given by
\begin{align}
\beta_{k,l,0} \sim \pi_{w^{***}} N_{R_{gen}} (0.8, 1) + (1 - \pi_{w^{***}}) \delta_{0}; \: k,l \in \{1,...,V\}.
\end{align}
Again, we set $\beta_{k,l,0}=\beta_{l,k,0}$ respecting the symmetry condition. If $\xi_k^0=0$ for any $k$, we set $\beta_{k,l,0}=\beta_{l,k,0}=0$ for any $l$. Contrary to \emph{Simulation 2}, \emph{Simulation 3} allows the possibility of an edge between the $k$th and the $l$th nodes having no impact on the response even when both $\xi_k^0$ and $\xi_l^0$ are nonzero. In the context of \emph{Simulation 3}, $(1-\pi_{w^{**}})$ and $(1-\pi_{w^{***}})$ are referred to as the \emph{node sparsity} and the \emph{edge sparsity} parameters respectively.

In the spirit of \emph{Simulation 2},  the network structure of nonzero coefficients is lost in \emph{Simulation 3} as well, and hence there is no inherent advantage of BNR. We present two cases of \emph{Simulation 3} in Table~\ref{Tabnew3}. Noise variance $\tau_0^2$ is fixed at $1$ for all scenarios. \\

\begin{table}[!th]
\begin{center}

\begin{tabular}
[c]{cccccc}
\hline
Cases & $R_{gen}$ & $R$ & Sparsity \\
\hline
Case - 1 & 2 & 2 & 0.5 \\
Case - 2 & 2 & 3 & 0.6 \\
Case - 3 & 2 & 5 & 0.3 \\
Case - 4 & 2 & 5 & 0.4 \\
Case - 5 & 3 & 5 & 0.5 \\
Case - 6 & 4 & 5 & 0.4 \\
Case - 7 & 2 & 5 & 0.5 \\
Case - 8 & 2 & 4 & 0.7\\
Case - 9 & 3 & 5 & 0.7\\
\hline
\end{tabular}
\caption{Table presents different cases for \emph{Simulation 1}. The true dimension $R_{gen}$ is the dimension of vector object $\bw_k$ using which data has been generated. The maximum dimension $R$ is the dimension of vector object $\bu_k$ using which the model has been fit. \emph{Sparsity} refers to the fraction of generated $\bw_k = \bzero$, i.e., $(1 - \pi_w)$.}\label{Tab1}
\end{center}
\end{table}

\begin{table}[!th]
\begin{center}
\begin{tabular}
[c]{cccc}
\hline
Cases & $R$ & Sparsity\\
\hline
Case - 1 & 5 & 0.7\\
Case - 2 & 5 & 0.2\\
\hline
\end{tabular}
\caption{Table presents different cases for \emph{Simulation 2}. The maximum dimension $R$ is the dimension of vector object $\bu_k$ using which the model has been fit. \emph{Simulation 2} only has a sparsity parameter $\pi_{w^*}$.}\label{Tabnew2}
\end{center}
\end{table}

\begin{table}[!th]
\begin{center}
\begin{tabular}
[c]{ccccc}
\hline
Cases & $R$ & Node Sparsity & Edge Sparsity\\
\hline
Case - 1 & 5 & 0.7 & 0.5\\
Case - 2 & 5 & 0.2 & 0.5\\
\hline
\end{tabular}
\caption{Table presents different cases for \emph{Simulation 3}. The maximum dimension $R$ is the dimension of vector object $\bu_k$ using which the model has been fit. While \emph{Simulation 2} only has a sparsity parameter, \emph{Simulation 3} has a node sparsity ($\pi_{w^{**}}$) and an edge sparsity ($\pi_{w^{***}}$) parameter respectively.}\label{Tabnew3}
\end{center}
\end{table}

\subsection{Results}\label{infff}
\noindent In all simulation results shown in this section, the BNR model is fitted with the choices of the hyper-parameters given by $\bS=\bI$, $\nu=10$, $a_{\Delta}=1$, $b_{\Delta}=1$,
$\zeta=1$ and $\iota=1$. 
Our extensive simulation studies reveal that both inference and prediction are robust to various choices of the hyper-parameters.\\

\noindent \underline{\emph{Node Selection}}\\
Figure \ref{Fig1} shows a matrix whose rows correspond to different cases in \emph{Simulation 1} and columns correspond to the nodes of the network. The dark and clear cells correspond to the truly active and inactive nodes, respectively. The posterior probability of the $k$th node being detected as active, i.e., $P(\xi_k = 1 \given Data)$  has been overlaid for all $k \in \{1,...20\}$ in all $9$ simulation cases. The plot suggests overwhelmingly accurate detection of nodes influencing the response. We provide similar plots on node detection in Figure~\ref{Fig1new} for various cases in \emph{Simulation 2} and \emph{Simulation 3}. Both figures show good detection of active nodes both in \emph{Simulations} \emph{2} and \emph{3} with very few false positives. We emphasize the fact that BNR is designed to detect important nodes while BLasso, HS or Lasso (or any other ordinary high dimensional regression technique) do not allow node selection in the present context. We have also investigated in detail the nodes selected by \citealp{relion2017network} and it is found to perform sub-optimally. For example, in Case -1 of \emph{Simulation 1}, out of $20$ nodes, the number of false positives is $9$.

\begin{figure}[!ht]
   \begin{center}
   \includegraphics[width=12 cm, height = 9 cm]{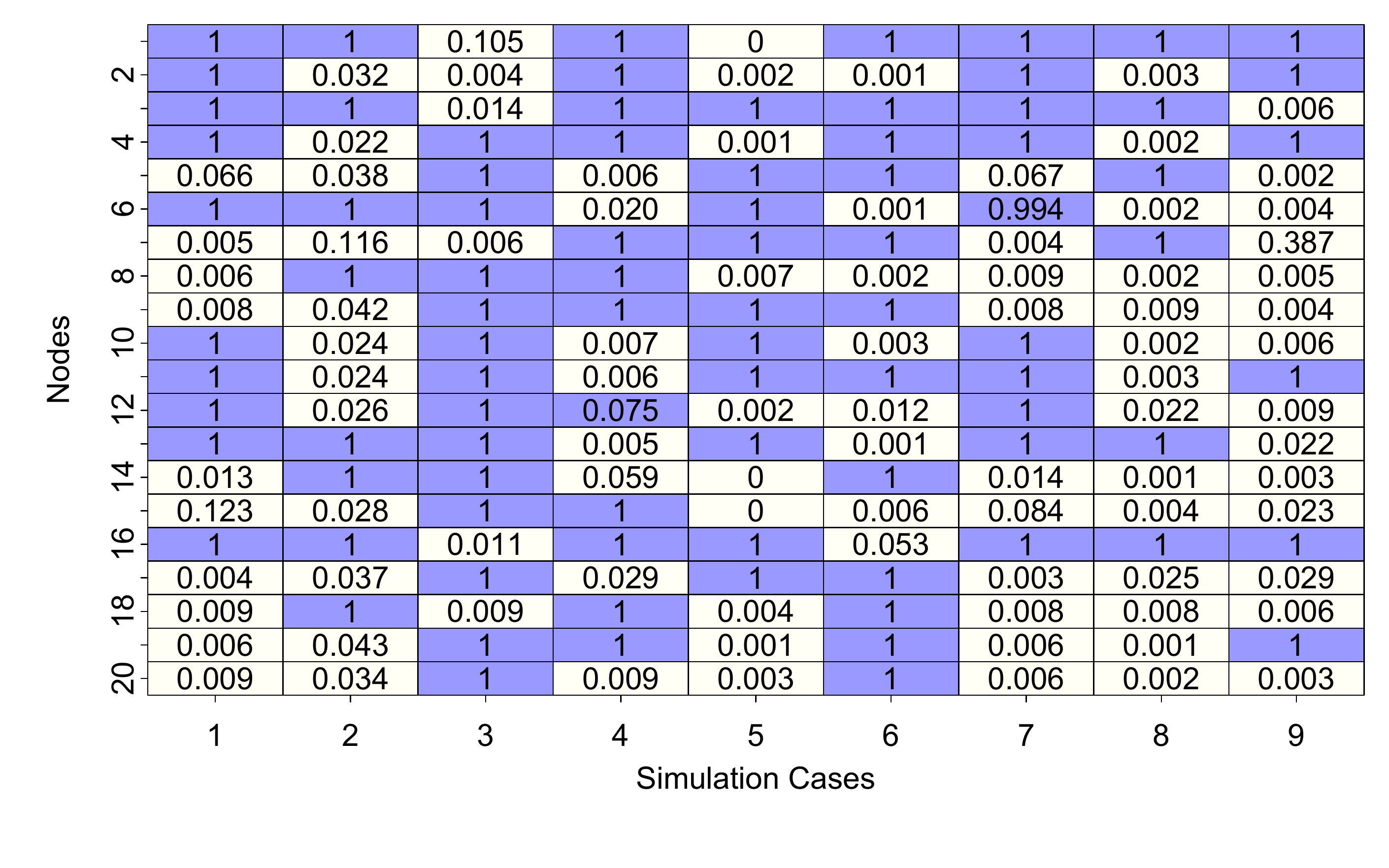}
   \end{center}
\caption{True activity status of a network node (clear blackground denotes \emph{inactive} and dark background denotes \emph{active}). Note that there are $20$ rows (corresponding to $20$ nodes) and $9$ columns corresponding to $9$ different cases in \emph{Simulation 1}. The model-detected posterior probability of being active has been super-imposed onto the corresponding node.}
\label{Fig1}
\end{figure}

\begin{figure}[!ht]
   \begin{center}
   \subfigure[Simulation 2]{\includegraphics[width=8 cm,height=9cm]{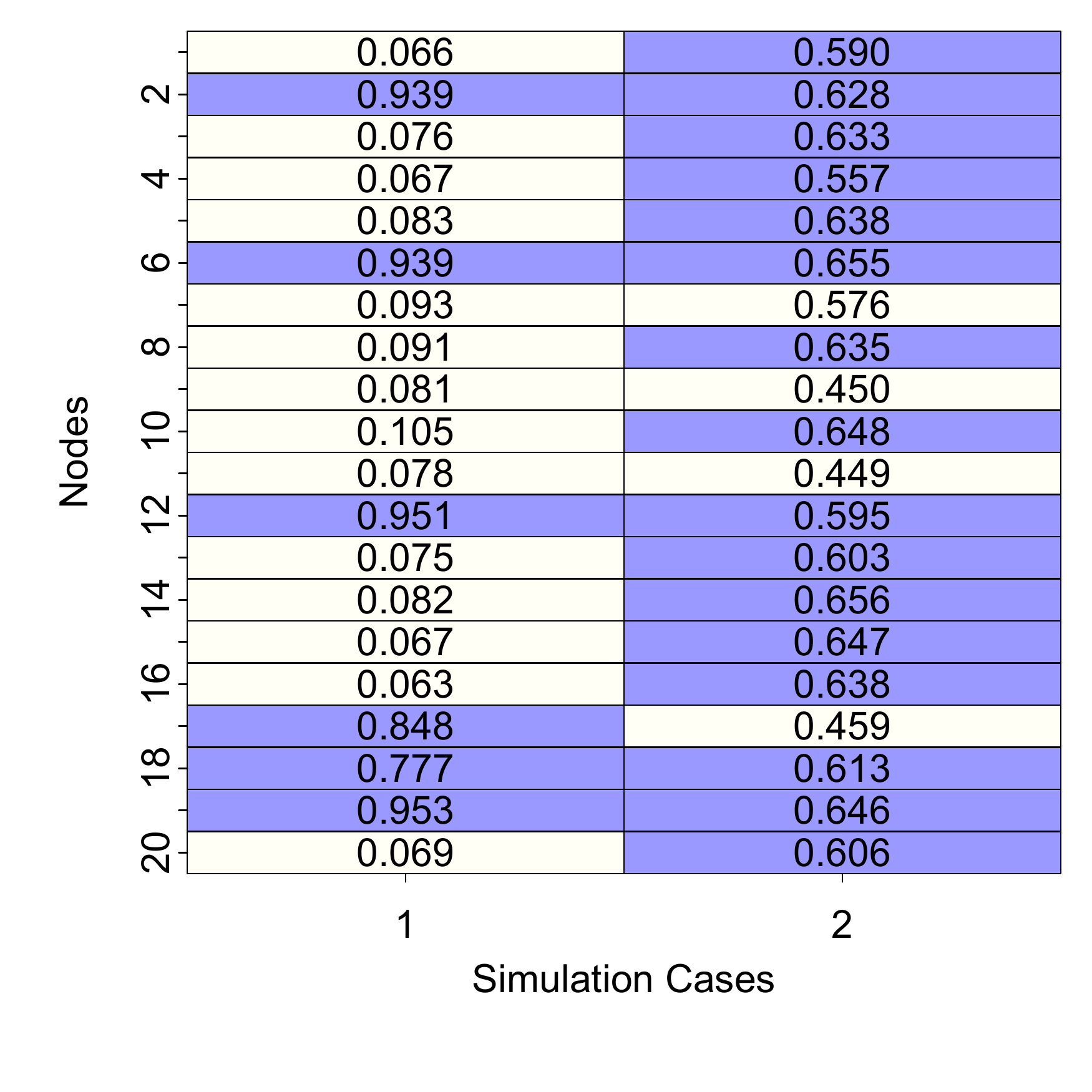}\label{corr_tpp111}}
   \subfigure[Simulation 3]{\includegraphics[width=8 cm,height=9cm]{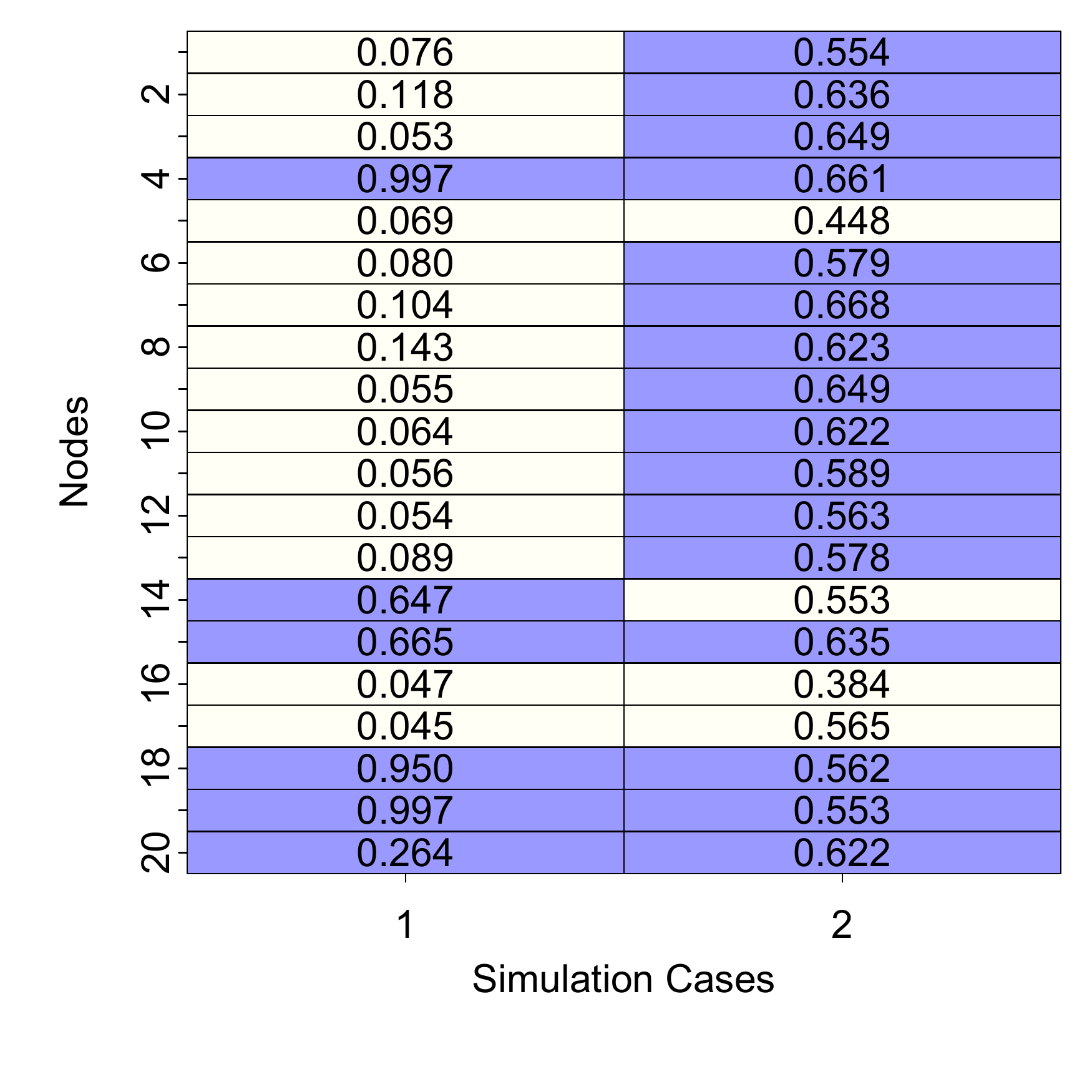}\label{corr_tpp2111}}
   \end{center}
\caption{True activity status of a network node (clear blackground denotes \emph{inactive} and dark background denotes \emph{active}). Note that there are $20$ rows (corresponding to $20$ nodes) and $2$ columns corresponding to $2$ different cases both in \emph{Simulations 2} and \emph{3}. The model-detected posterior probability of being active has been super-imposed onto the corresponding node.}
\label{Fig1new}
\end{figure}

\noindent \underline{\emph{Estimating the network coefficient}}\\
Another important aspect of the parametric inference lies in the performance in terms of estimating the network predictor coefficients. Tables~\ref{Tab2}, \ref{Tab2new}, \ref{Tab2newnew} present the MSE of all the competitors in \emph{Simulations} \emph{1}, \emph{2} and \emph{3} respectively. Given that both the fitted network regression coefficient $\bB$ and the true coefficient $\bB_0$ are symmetric, the MSE for any competitor is calculated in each dataset as $\frac{1}{(V(V-1)/2)}\sum_{k<l}(\hat{\gamma}_{k,l}-\gamma_{k,l,0})^2$, where $\hat{\gamma}_{k,l}$ is the point estimate of $\gamma_{k,l} = 2\beta_{k,l}$, $\gamma_{k,l,0}$ is the true value of the coefficient, and $\frac{V(V-1)}{2}$ is the number of elements in the upper triangular portions of the $V \times V$ matrices $\bB$ and $\bB_0$. 
For Bayesian models (such as the proposed model), $\hat{\gamma}_{k,l}$ is taken to be the posterior mean of $\gamma_{k,l}$.  

Table~\ref{Tab2} shows that the proposed Bayesian network regression (BNR) outperforms all its Bayesian and frequentist competitors in all cases of \emph{Simulation 1}. In cases 1-7 when the sparsity parameter is low to moderate, we perform overwhelmingly better than all the competitors. While BNR is expected to perform much better than BLasso, Horseshoe and Lasso due to incorporation of network information, it is important to note that the carefully chosen global-local shrinkage prior with a well formulated hierarchical mean structure seems to possess more detection and estimation power than \citealp{relion2017network}. When the sparsity parameter in \emph{Simulation 1}  is high (cases 8-9), our simulation scheme sets an overwhelming proportion of $\beta_{k,l,0}$'s to zero. As a result, BNR only slightly outperforms Horseshoe and BLasso. Again, \citealp{relion2017network} does not seem to be competitive, not only against BNR, but also against Horseshoe. 
For \emph{Simulations 2} and \emph{3}, tables~\ref{Tab2new} and \ref{Tab2newnew} demonstrate that when node sparsity is higher (i.e. when more edge coefficients are set to zero in the data generation procedure), BNR performs very similar to Horseshoe. This might be due to the fact that high degree of sparsity in the edge coefficients in the truth is conducive for ordinary high dimensional regression which treats edges as coefficients. As node sparsity decreases so that more and more edge coefficients are nonzero in the truth and the network structure in the predictors dominates, BNR tends to show more and more advantage in terms of estimating the network coefficient $\bB$.\\

\noindent \emph{\underline{Inference on the effective dimensionality}}
\newline Next, the attention turns to inferring on the posterior expected value of the effective dimensionality of $\bu_k$. Figure~\ref{effective} presents posterior probabilities of effective dimensionality in all $9$ cases in \emph{Simulation 1}. The filled bullets indicate the true value of the effective dimensionality. All 9 figures indicate that the true dimensionality of the latent variable $\bu_k$ is effectively captured by the models. Variation in the sparsity or discrepancy between $R$ and $R_{gen}$ seems to impact the performance very negligibly. Only in cases 8 and 9, in presence of less network information, there appears to be greater uncertainty in the posterior distribution of $R_{eff}$. We also investigate similar figures for \emph{Simulations 2} and \emph{3} (see Figures~\ref{effective2} and \ref{effective3}), though in the absence of any ground truth on effective dimensionality, they are less interpretable. 

\begin{figure}
  \begin{center}
    \subfigure[Case 1]{\includegraphics[width=5 cm,height=5cm]{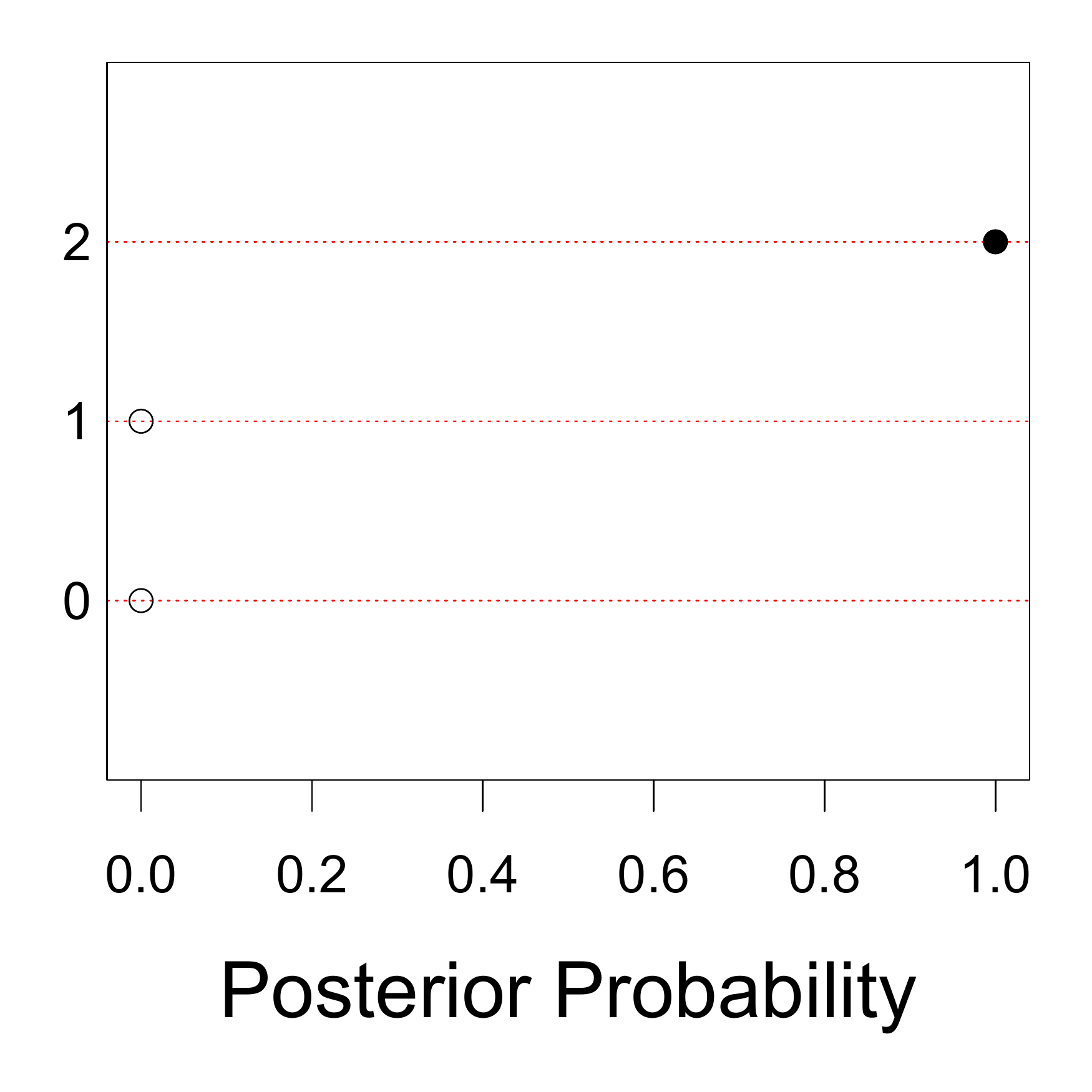}\label{corr_tpp1}}
   \subfigure[Case 2]{\includegraphics[width=5 cm,height=5cm]{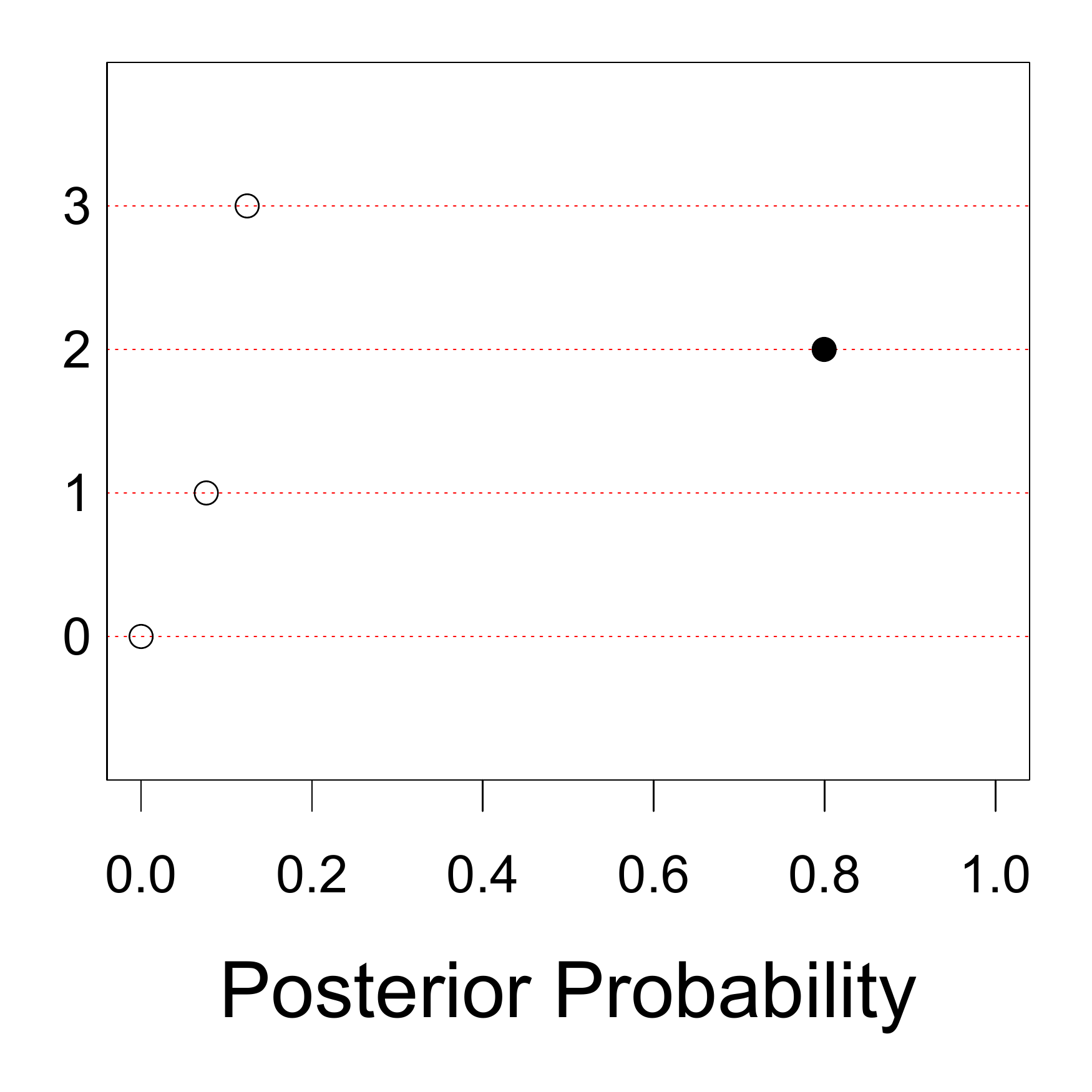}\label{corr_tpp2}}
   \subfigure[Case 3]{\includegraphics[width=5 cm,height=5cm]{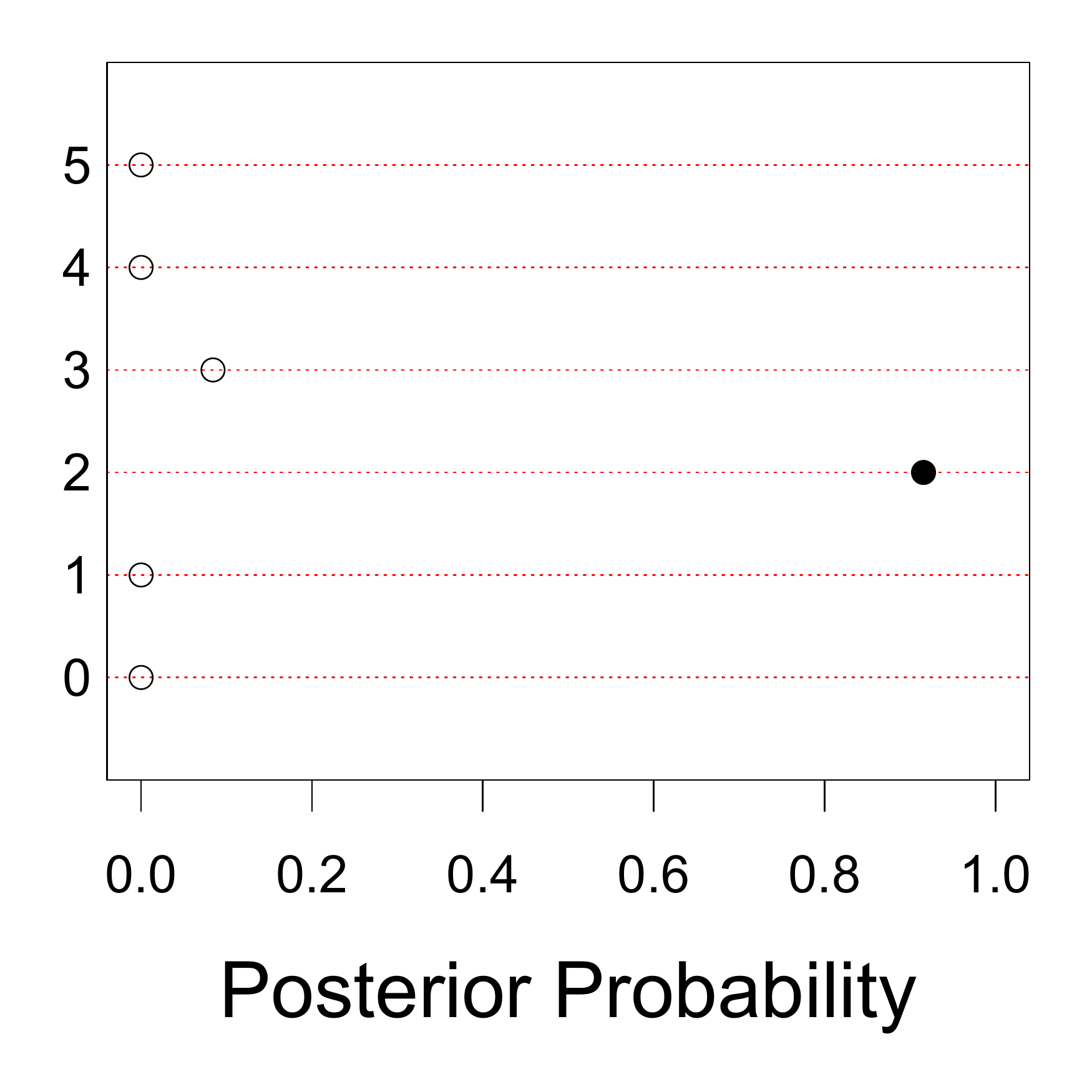}\label{corr_tppp1}}\\
   \subfigure[Case 4]{\includegraphics[width=5 cm,height=5cm]{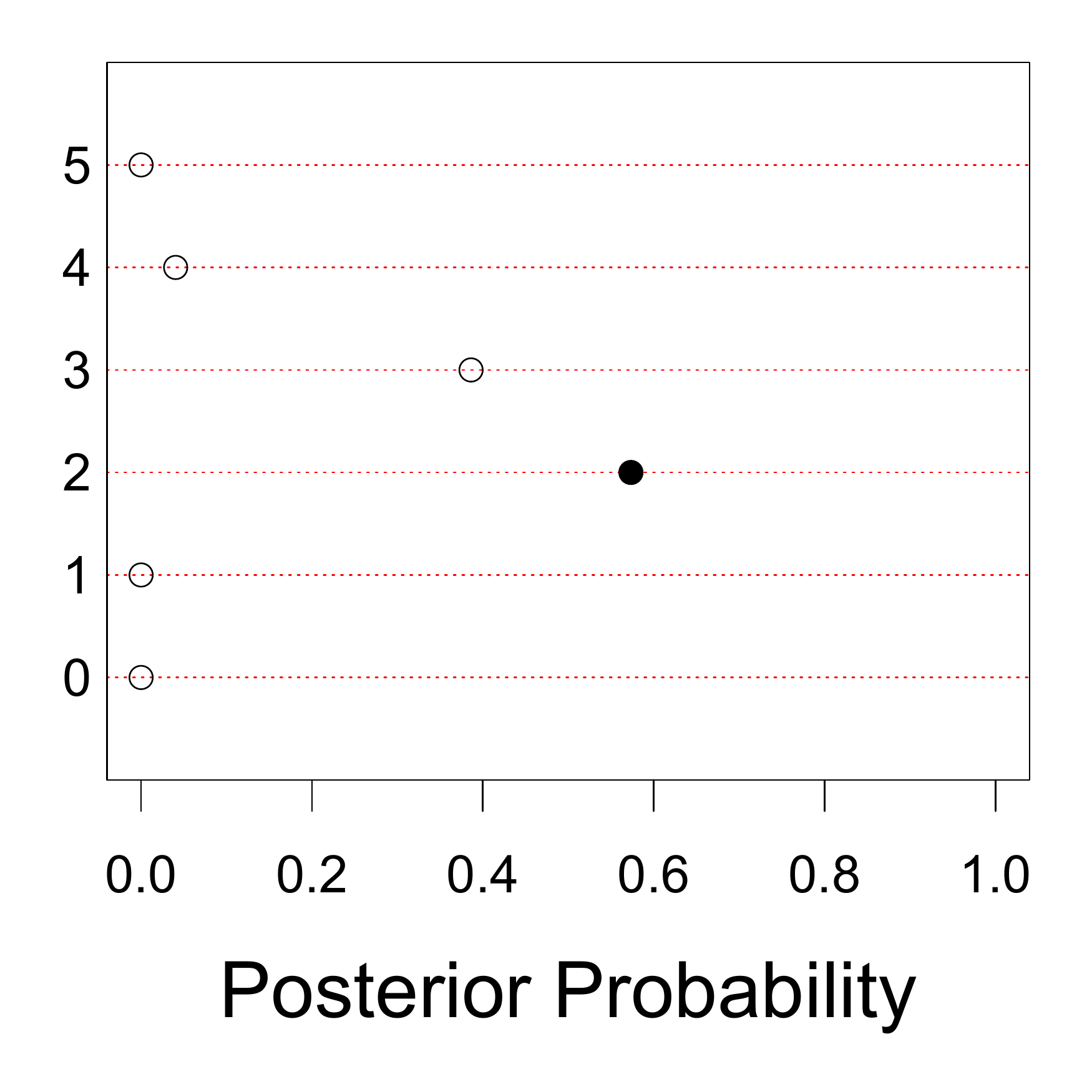}\label{corr_tppp2}}
   \subfigure[Case 5]{\includegraphics[width=5 cm,height=5cm]{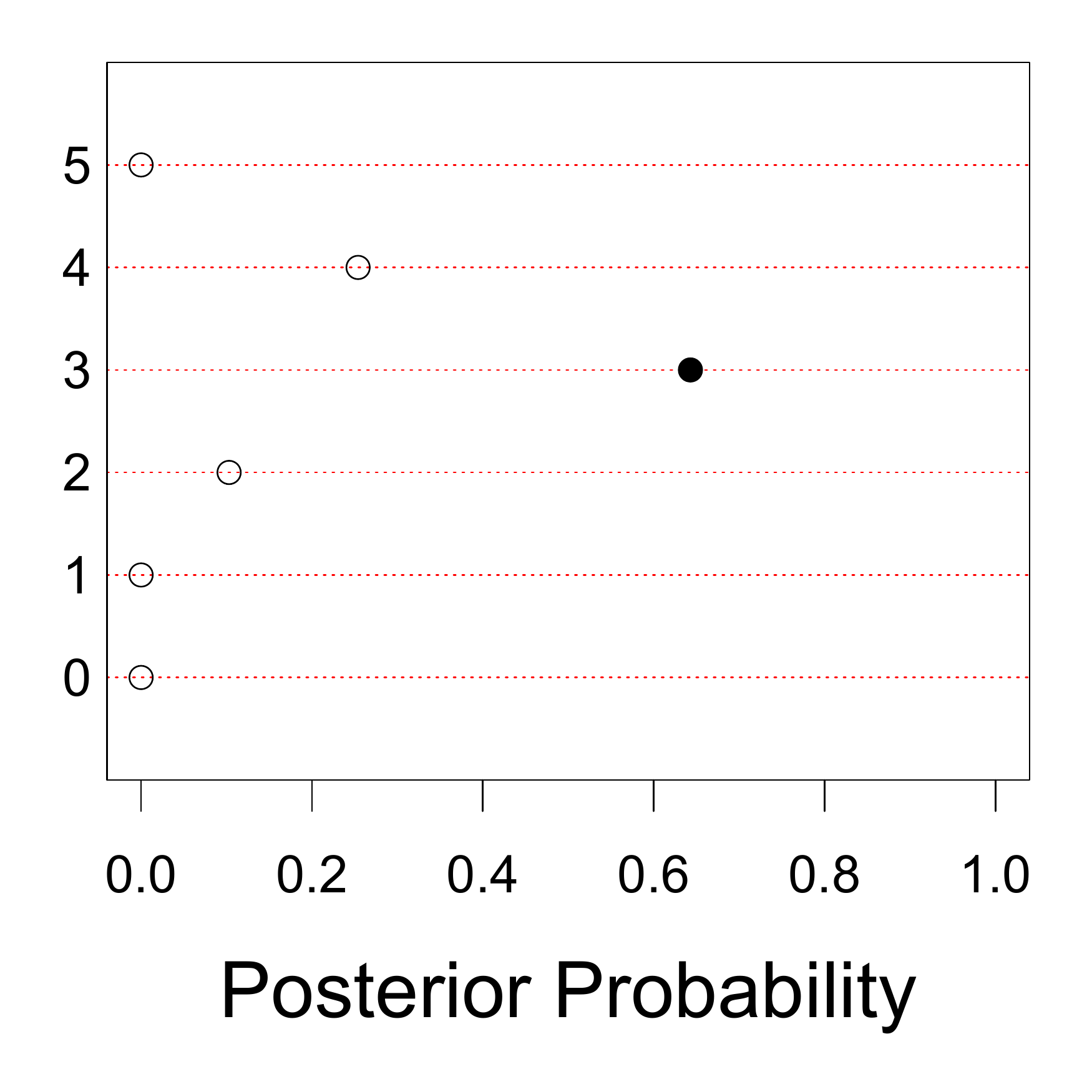}\label{corr_mpp1}}
   \subfigure[Case 6]{\includegraphics[width=5 cm,height=5cm]{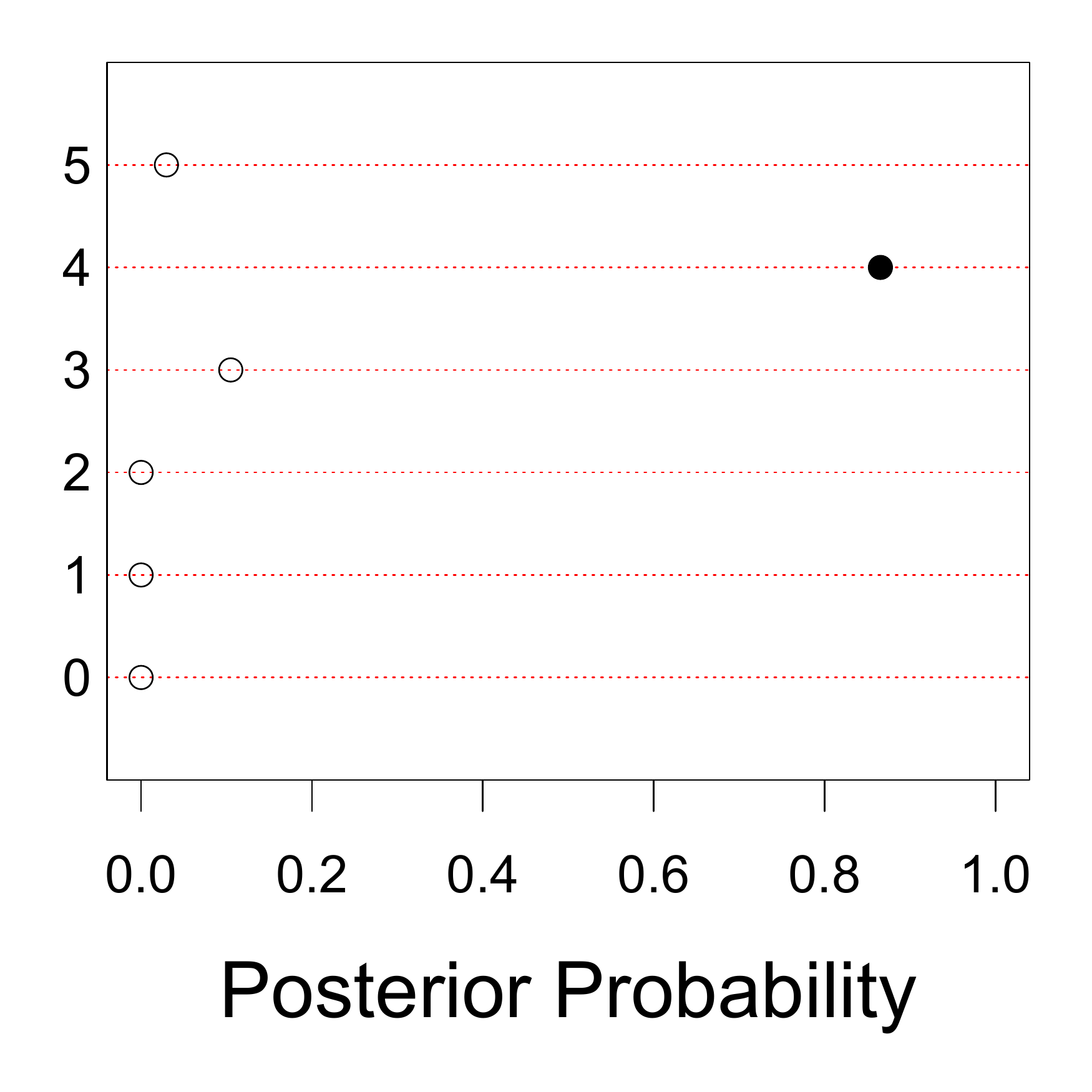}\label{corr_mpp2}}\\
   \subfigure[Case 7]{\includegraphics[width=5 cm,height=5cm]{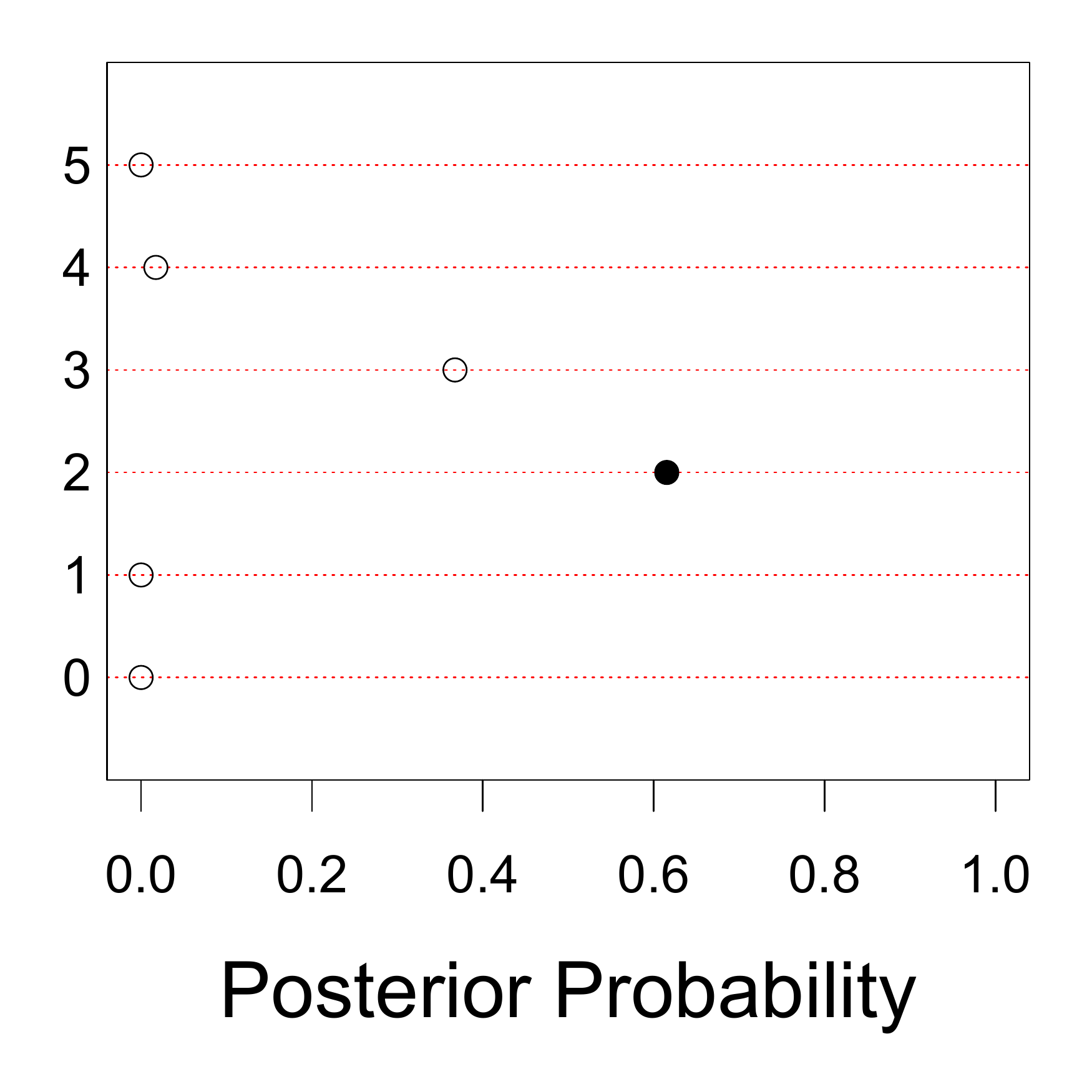}\label{corr2}}
   \subfigure[Case 8]{\includegraphics[width=5 cm,height=5cm]{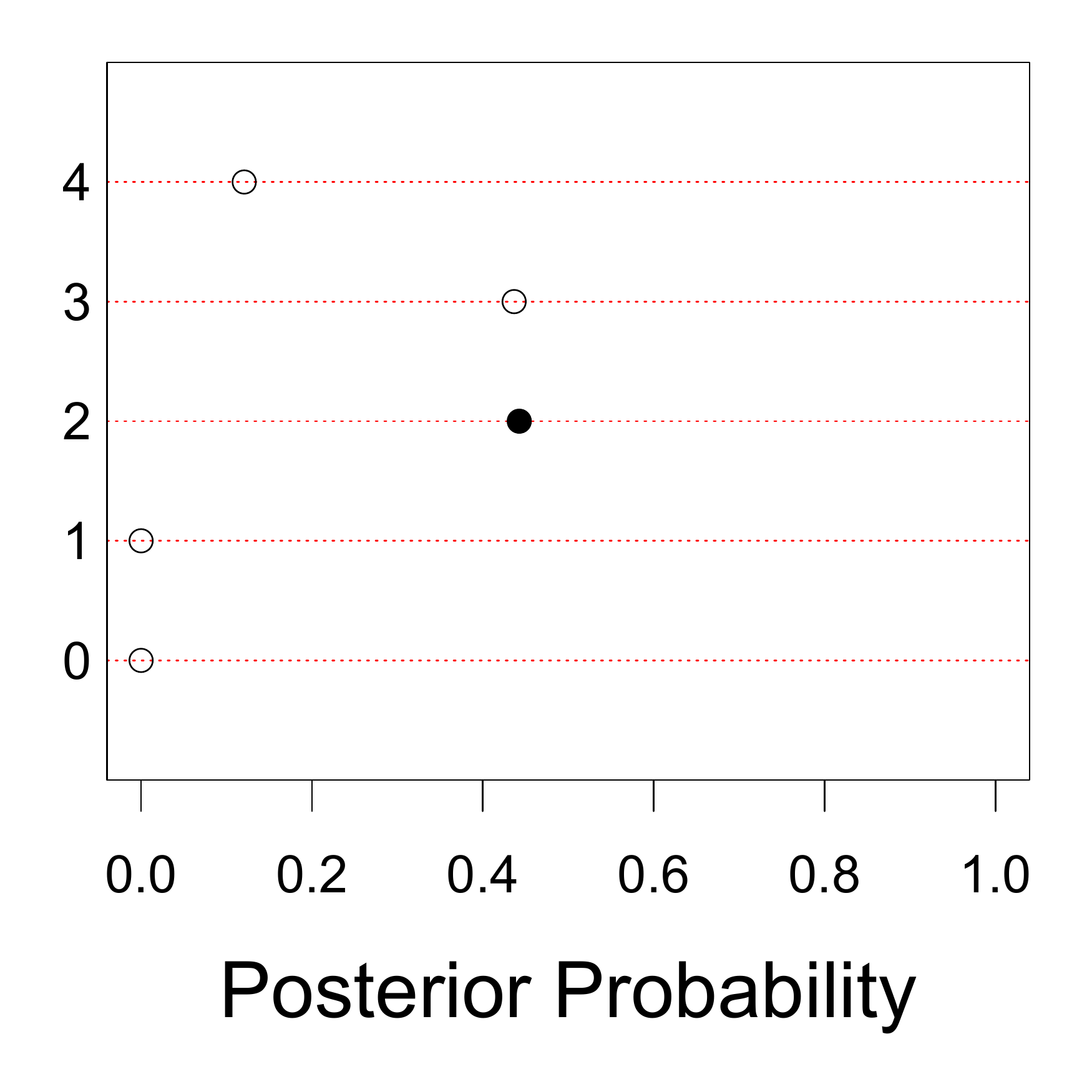}\label{corr_mpp3}}
   \subfigure[Case 9]{\includegraphics[width=5 cm,height=5cm]{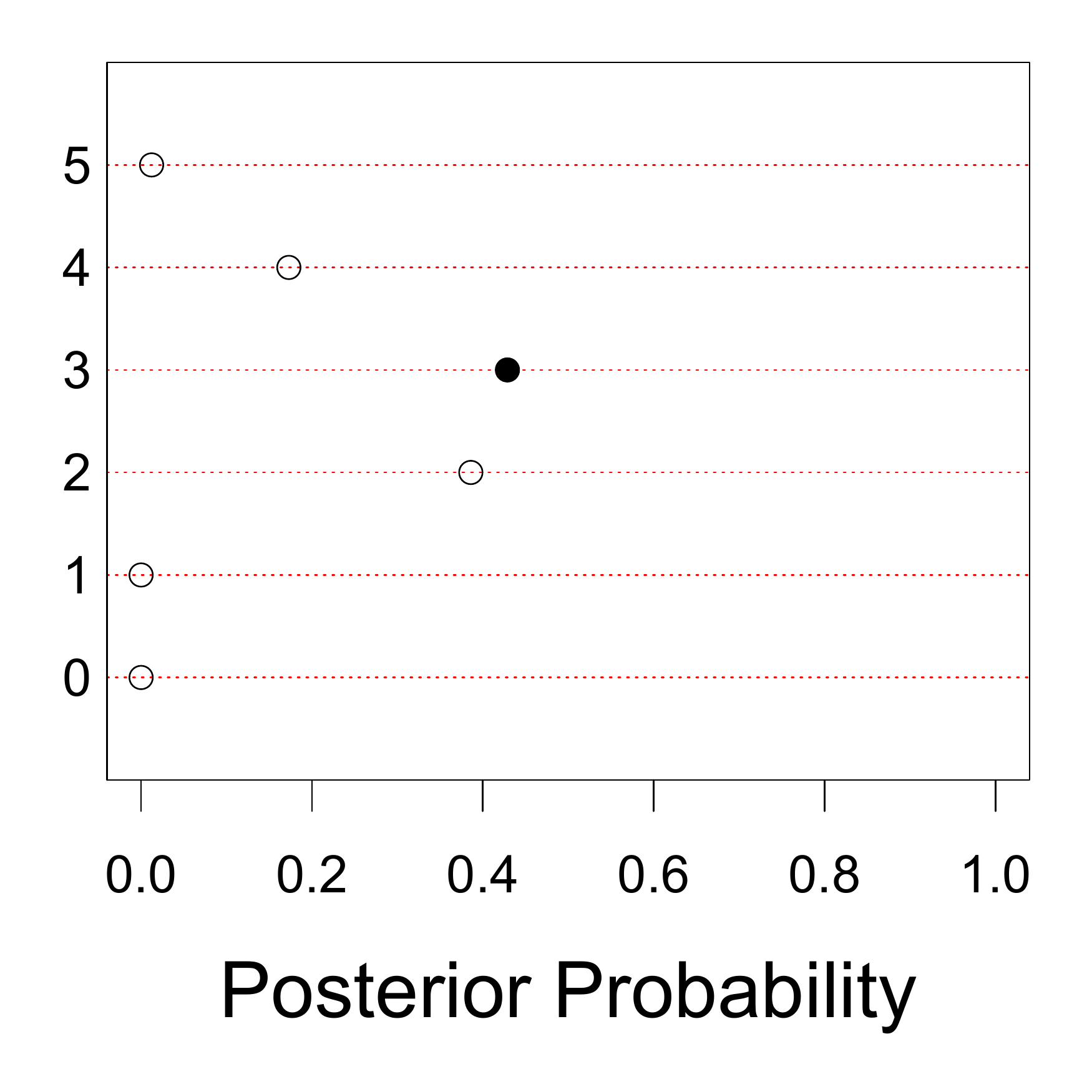}\label{corr3}}
 \end{center}
 \caption{Plots showing posterior probability distribution of effective dimensionality in all $9$ cases in \emph{Simulation 1}. Filled bullets indicate the true value of effective dimensionality.}\label{effective}
\end{figure}

\begin{figure}
  \begin{center}
    \subfigure[Case 1]{\includegraphics[width=5 cm,height=5cm]{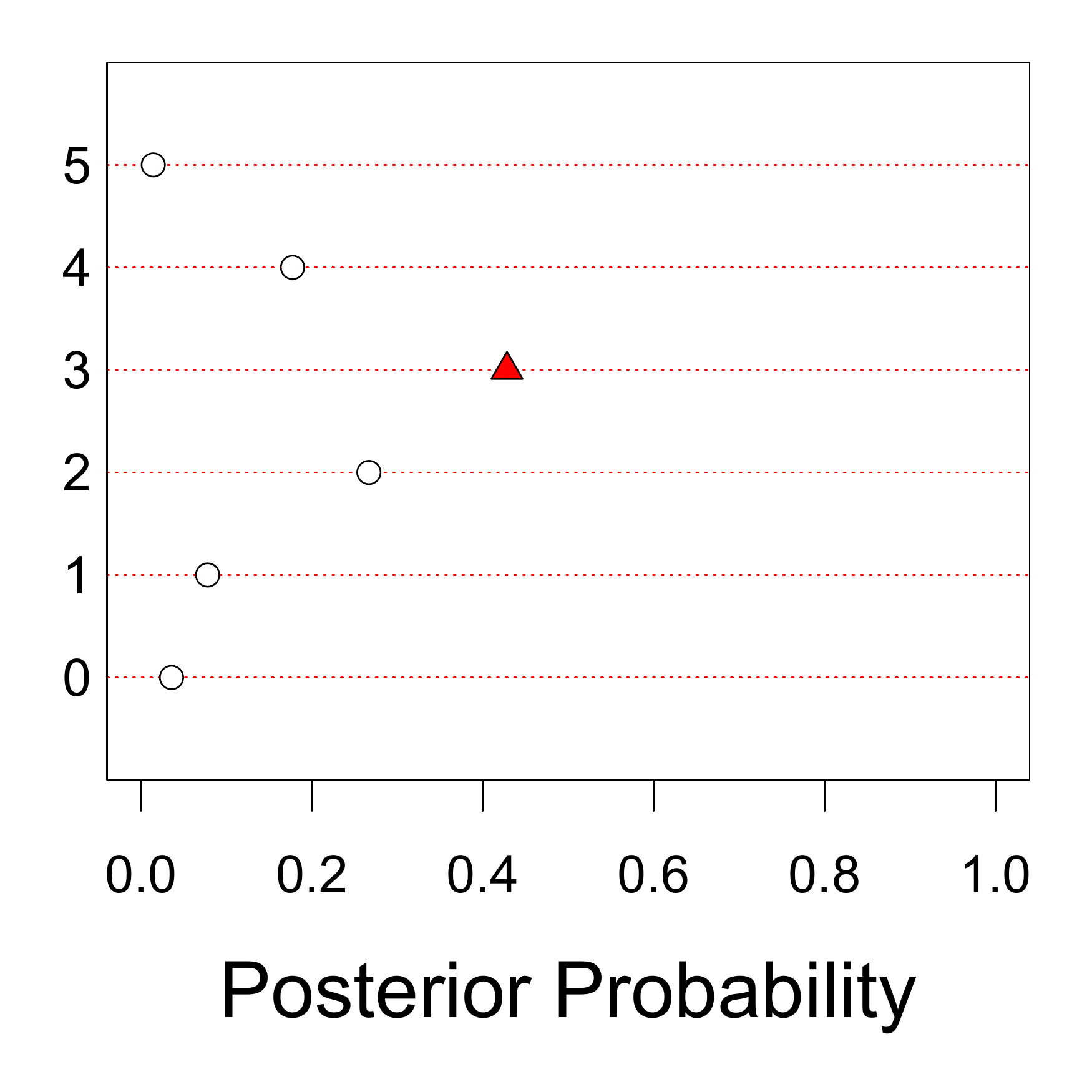}\label{sim2case1dot}}
   \subfigure[Case 2]{\includegraphics[width=5 cm,height=5cm]{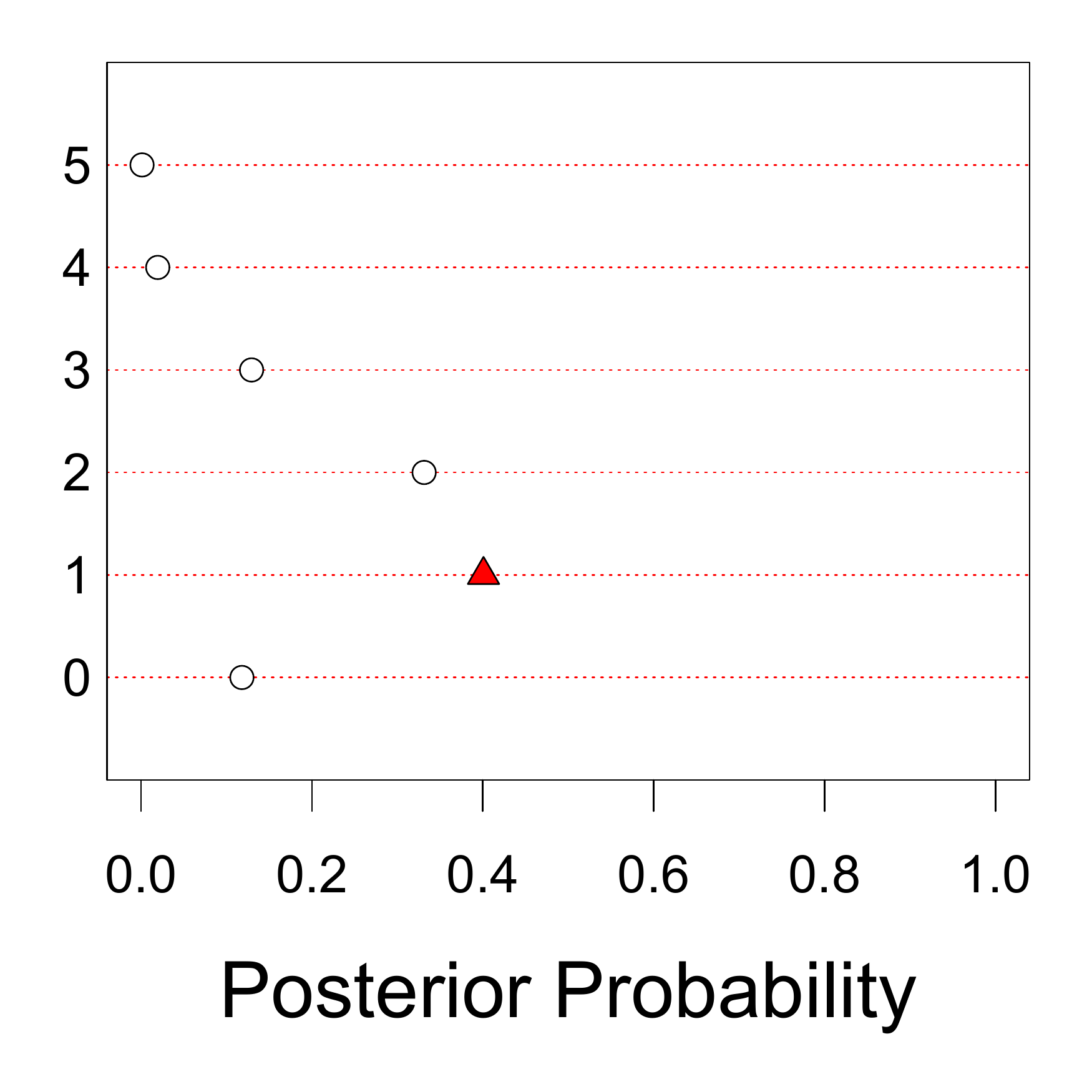}\label{sim2case1dot}}
 \end{center}
 \caption{Plots showing posterior probability distribution of effective dimensionality in the $2$ cases in \emph{Simulation 2}.}\label{effective2}
\end{figure}

\begin{figure}
  \begin{center}
    \subfigure[Case 1]{\includegraphics[width=5 cm,height=5cm]{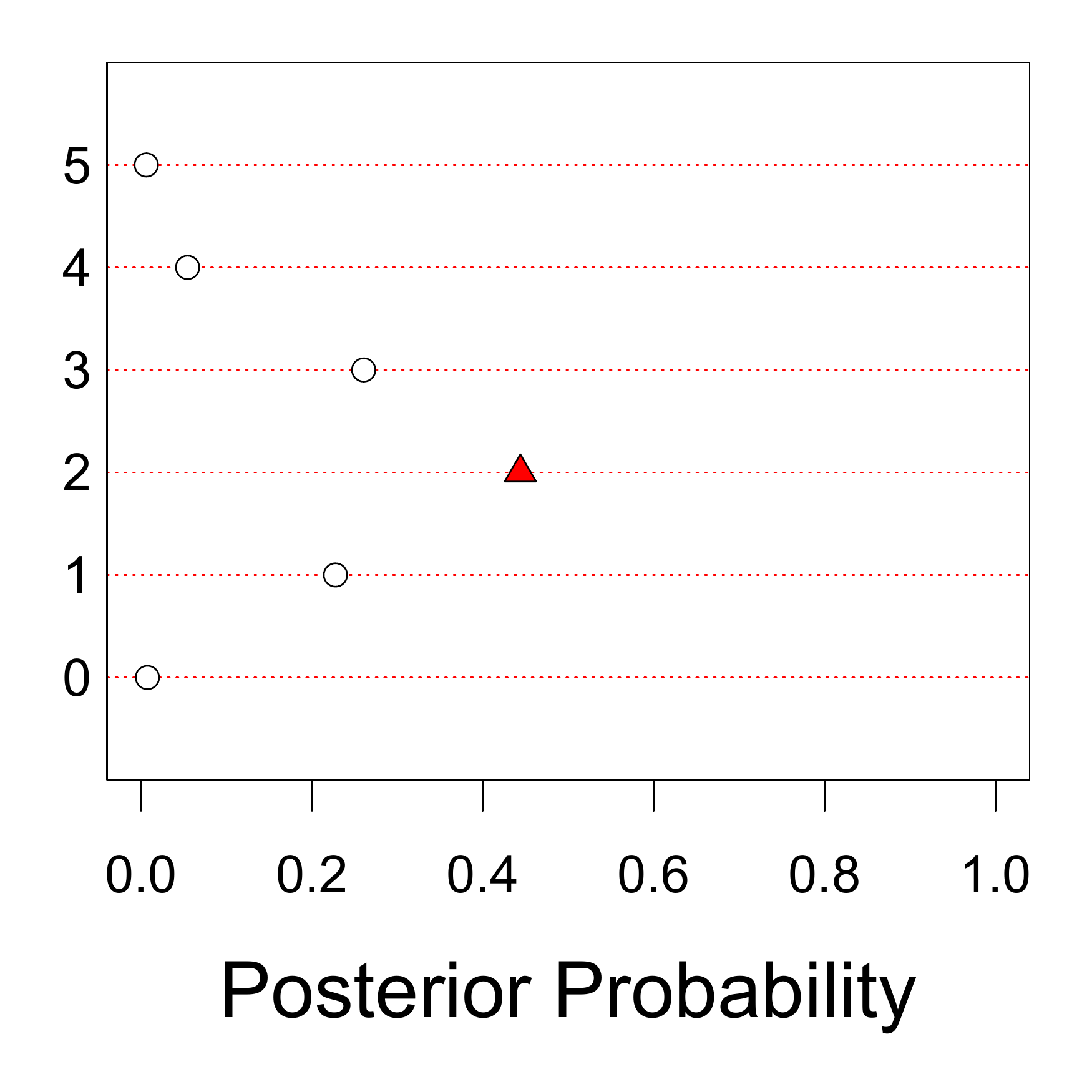}\label{sim3case1dot}}
   \subfigure[Case 2]{\includegraphics[width=5 cm,height=5cm]{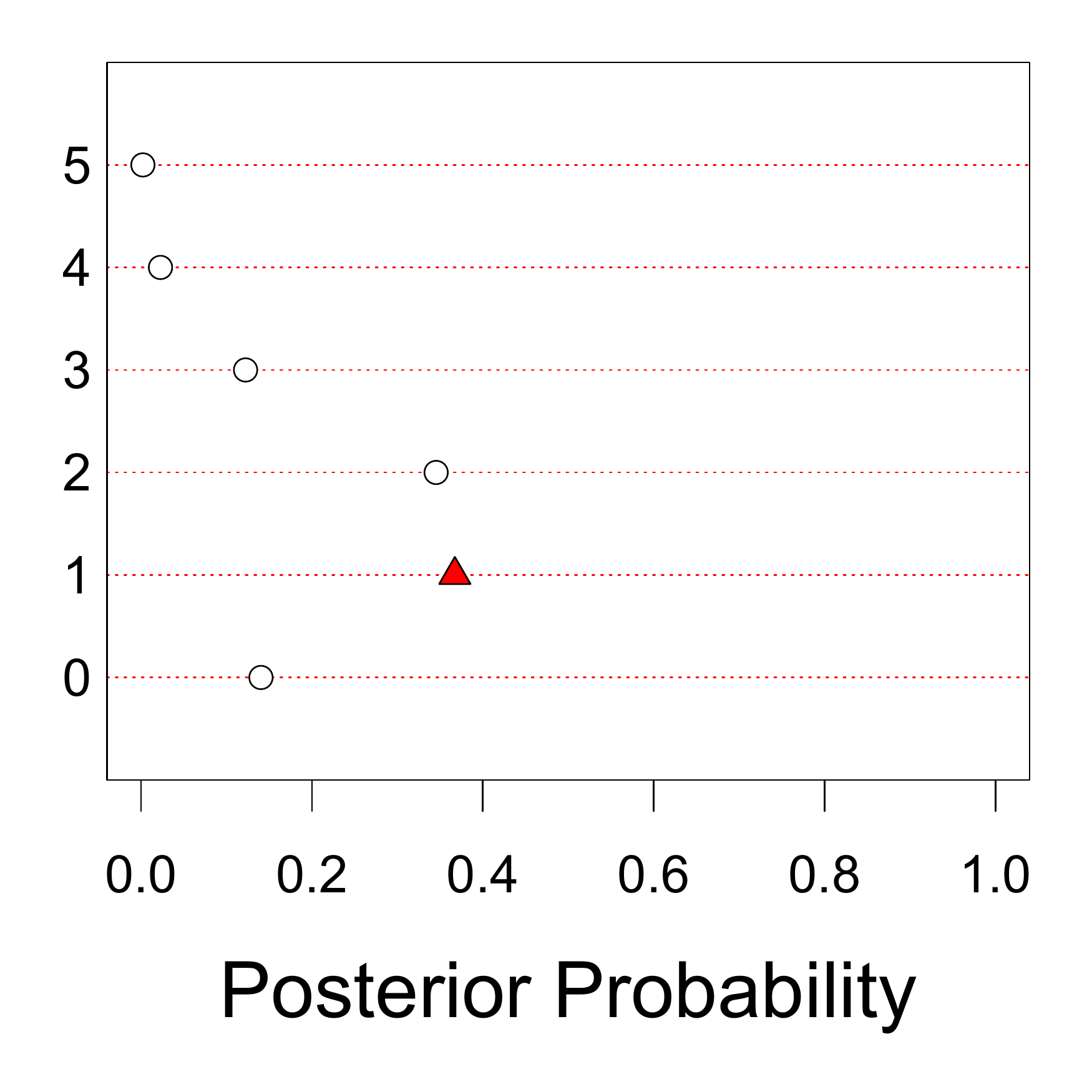}\label{sim3case1dot}}
 \end{center}
 \caption{Plots showing posterior probability distribution of effective dimensionality in the $2$ cases in \emph{Simulation 3}.}\label{effective3}
\end{figure}


\begin{table}[!th]
\begin{center}

\begin{tabular}
[c]{cccc|ccccc}
\hline
\multicolumn{4}{c}{} & \multicolumn{5}{|c}{MSE}\\
\hline
Cases & $R_{gen}$ & $R$ & Sparsity & \textbf{BNR} & Lasso &  Reli\'{o}n(2017) & BLasso & Horseshoe \\
\hline
Case - 1 & 2 & 2 & 0.5 & \textbf{0.009} & 0.438 & 0.524 & 0.472 &  0.395\\
Case - 2 & 2 & 3 & 0.6  & \textbf{0.007} & 0.660 & 0.929 &  0.863 & 0.012\\
Case - 3 & 2 & 5 & 0.3  & \textbf{0.006} & 1.295 & 1.117 &  1.060 & 1.070\\
Case - 4 & 2 & 5 & 0.4 &  \textbf{0.006} & 0.371 & 0.493 &  0.699 & 0.298\\
Case - 5 & 3 & 5  & 0.5  & \textbf{0.009} & 1.344 & 1.629 & 1.638 & 1.381\\
Case - 6 & 4 & 5 & 0.4 & \textbf{0.006} & 3.054 & 2.601 & 2.680 & 3.284\\
Case - 7 & 2 & 4 & 0.5  & \textbf{0.009} & 0.438 & 0.524 & 0.472 &  0.395\\
Case - 8 & 2 & 4 & 0.7 & \textbf{0.005} & 0.015 & 0.251 & 0.007 & 0.008\\
Case - 9 & 3 & 5 & 0.7  & \textbf{0.004} & 0.029 & 0.071 & 0.019 &  0.007\\
\hline
\end{tabular}
\caption{Performance of Bayesian Network Regression (BNR) vis-a-vis competitors for cases in \emph{Simulation 1}. Parametric inference in terms of point estimation of edge coefficients has been captured through the Mean Squared Error (MSE). The minimum MSE among competitors for any case is made bold.}\label{Tab2}
\end{center}
\end{table}

\begin{table}[!th]
\begin{center}
\begin{tabular}
[c]{cccc|ccccc}
\hline
\multicolumn{4}{c}{} & \multicolumn{5}{|c}{MSE}\\
\hline
Cases & $R_{gen}$ & $R$ & Sparsity & \textbf{BNR} & Lasso &  Reli\'{o}n(2017) & BLasso & Horseshoe \\
\hline
Case - 1 & 3 & 5 & 0.7 & 0.011 & 0.013 & 0.036 & 0.010 &  \textbf{0.008}\\
Case - 2 & 3 & 5 & 0.2  & \textbf{0.629} & 0.843 & 0.859 &  0.836 & 0.948\\
\hline
\end{tabular}
\caption{Performance of Bayesian Network Regression (BNR) vis-a-vis competitors for cases in \emph{Simulation 2}. Parametric inference in terms of point estimation of edge coefficients has been captured through the Mean Squared Error (MSE). The minimum MSE among competitors for any case is made bold.}\label{Tab2new}
\end{center}
\end{table}

\begin{table}[!th]
\begin{center}

\begin{tabular}
[c]{ccccc|ccccc}
\hline
\multicolumn{5}{c}{} & \multicolumn{5}{|c}{MSE}\\
\hline
Cases & $R_{gen}$ & $R$ & Node & Edge & \textbf{BNR} & Lasso &  Reli\'{o}n(2017) & BLasso & Horseshoe \\
      &           &     & Sparsity & Sparsity &      &       &                   &        & \\
\hline
Case - 1 & 3 & 5 & 0.7 & 0.5 & 0.004 & 0.006 & 0.017 & 0.004 &  \textbf{0.003}\\
Case - 2 & 3 & 5 & 0.2 & 0.5  & \textbf{0.457} & 0.636 & 0.617 &  0.659 & 0.629\\
\hline
\end{tabular}
\caption{Performance of Bayesian Network Regression (BNR) vis-a-vis competitors for cases in \emph{Simulation 3}. Parametric inference in terms of point estimation of edge coefficients has been captured through the Mean Squared Error (MSE). The minimum MSE among competitors for any case is made bold.}\label{Tab2newnew}
\end{center}
\end{table}

\subsection{Predictive Inference}
For the purpose of assessing predictive inference of the competitors, $n_{pred}=30$ samples are generated from (\ref{data_gen}). We compare the predictive ability of competitors based on the point prediction and characterization of predictive uncertainties. To assess point prediction, we employ the mean squared prediction error (MSPE) which is obtained as the average squared distance between the point prediction and the true responses for all the competitors. As measures of predictive uncertainty, we provide coverage and length of $95\%$ predictive intervals. For frequentist competitors, $95\%$ predictive intervals are obtained by using predictive point estimates plus and minus $1.96$ times standard errors.
Figure~\ref{Fig_pred} provides all three measures for all competitors in the $9$ cases for \emph{Simulation 1}.

It is quite evident from Figure~\ref{MSPE} that BNR outperforms other competitors in terms of point prediction. Among the competitors, Horseshoe does a reasonably good job in cases with a higher degree of sparsity. Note that the data generation procedure in \emph{Simulation 1} ensures that the elements in the edge coefficient matrix $\bB_0$ take reasonably large positive and negative values. Additionally, an overwhelming number of edge coefficients is set to zero. Clearly, with a small sample size, all local and global parameters for the ordinary vector shrinkage priors are shrunk to zero to provide a good estimation of zero coefficients. While doing so, they miss out on estimating coefficients which significantly deviate from zero. Thus, the ordinary vector shrinkage priors do a relatively poor job in terms of point prediction. Lasso also fails to provide satisfactory performance for similar reasons. On the other hand, all models provide close to nominal coverage. However, the predictive intervals associated with BNR are much narrower than those of other competitors. 
\begin{figure}
  \begin{center}
    \subfigure[MSPE]{\includegraphics[width=5.0 cm]{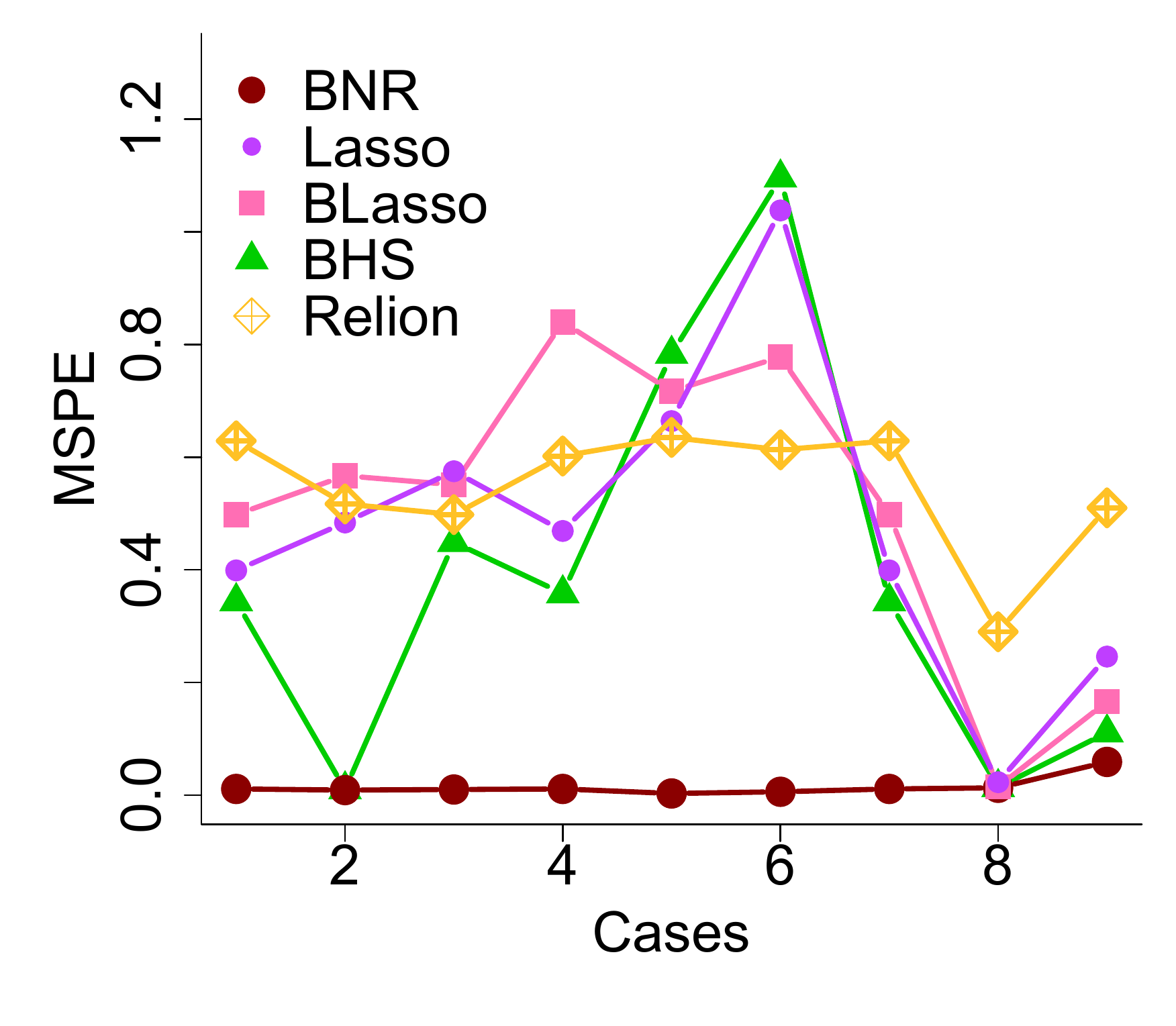}\label{MSPE}}
    \subfigure[Coverage of 95\% PI]{\includegraphics[width=5.0 cm]{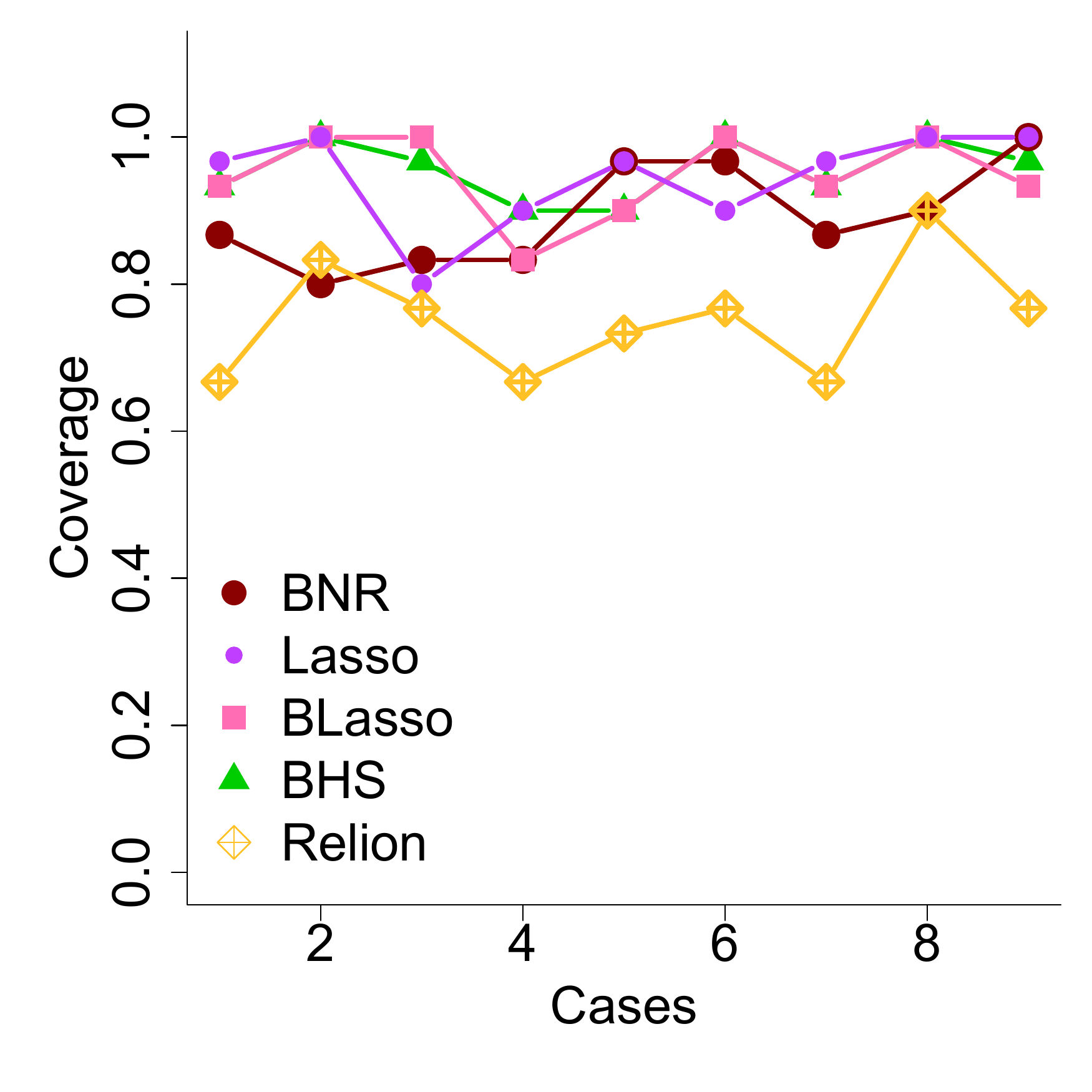}\label{Coverage}}
    \subfigure[Length of 95\% PI]{\includegraphics[width=5.0 cm]{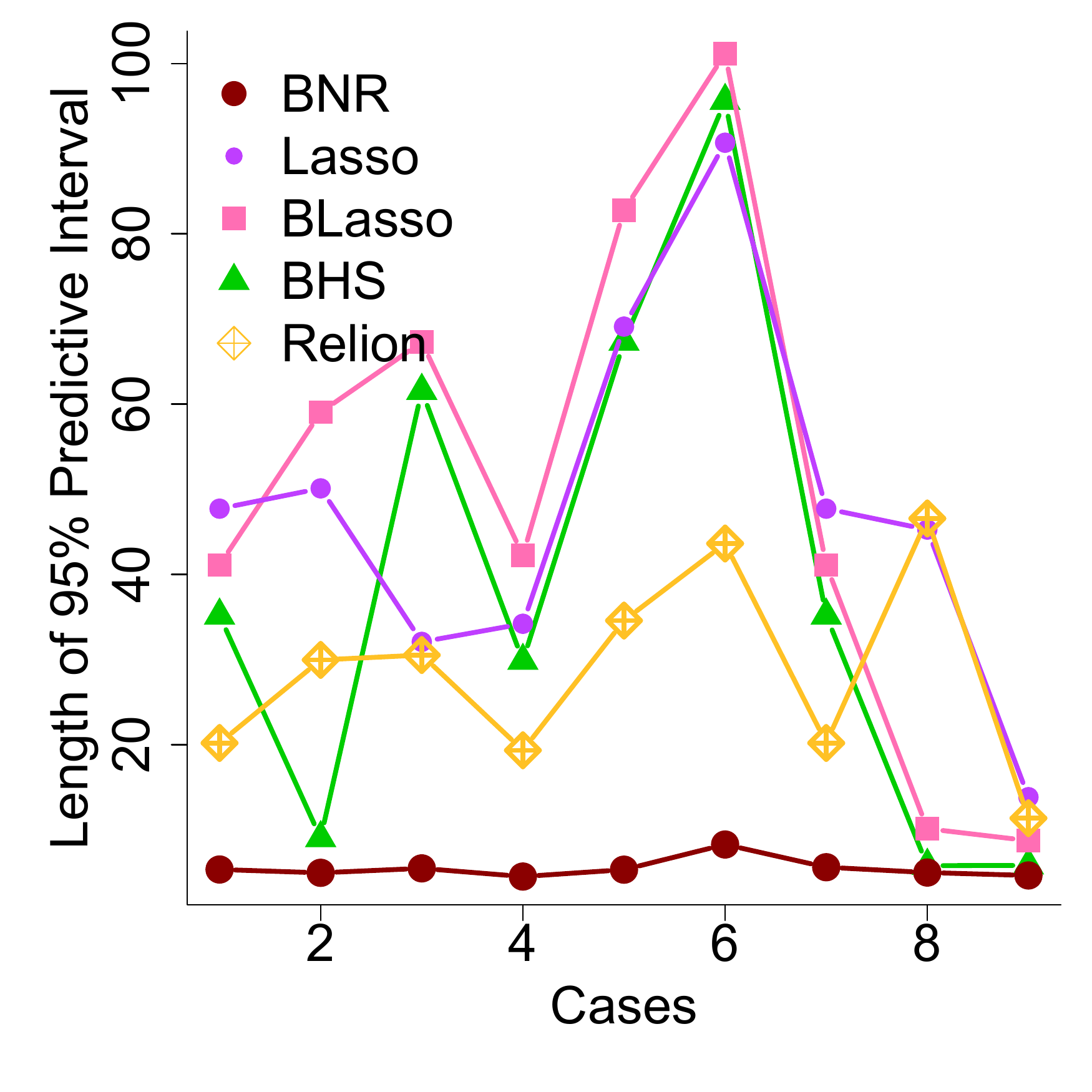}\label{Length}}
 \end{center}
 \caption{Figures from left to right show MSPE, coverage and length of 95\% predictive intervals for all competitors.}\label{Fig_pred}
\end{figure}

The predictive performance for \emph{Simulations 2} and \emph{3} is given in Tables~\ref{Tab3pred} and \ref{Tab3prednew}
respectively. Since the data generation schemes in \emph{Simulations 2} and \emph{3} do not incorporate network structure, the performance of BNR is comparable to that of its competitors in the high sparsity case. On the other hand, in presence of low node sparsity, the performance of Lasso, BLasso and Horseshoe deteriorate, and BNR turns out to be the best performer. 
\begin{table}[!th]
\begin{center}
\begin{tabular}
[c]{cccc|ccccc}
\hline
\multicolumn{4}{c}{} & \multicolumn{5}{|c}{MSPE}\\
\hline
Cases & $R_{gen}$ & $R$ & Sparsity & \textbf{BNR} & Lasso &  Reli\'{o}n(2017) & BLasso & Horseshoe \\
\hline
Case - 1 & 3 & 5 & 0.7 & 0.079 & 0.100 & 0.371 & 0.076 &  \textbf{0.061}\\
Case - 2 & 3 & 5 & 0.2  & \textbf{0.432} & 0.726 & 0.859 &  0.629 & 0.725\\
\hline
\multicolumn{4}{c}{} & \multicolumn{5}{|c}{Coverage of 95\% PI}\\
\hline
Case - 1 & 3 & 5 & 0.7 & 1.00 & 1.00 & 0.867 & 1.00 &  0.967\\
Case - 2 & 3 & 5 & 0.2  & 0.94 & 0.73 & 0.56 &  0.96 & 0.87\\
\hline
\multicolumn{4}{c}{} & \multicolumn{5}{|c}{Length of 95\% PI}\\
\hline
Case - 1 & 3 & 5 & 0.7 & 8.97 & 18.70 & 10.23 & 8.25 &  6.40\\
Case - 2 & 3 & 5 & 0.2  & 42.18 & 32.81 & 23.54 &  45.69 & 42.91\\
\hline
\end{tabular}
\caption{MSPE, coverage and length of 95\% predictive intervals (PIs) of Bayesian Network Regression (BNR) vis-a-vis competitors for cases in \emph{Simulation 2}. Lowest MSPE for any case is made bold.}\label{Tab3pred}
\end{center}
\end{table}

\begin{table}[!th]
\begin{center}
\begin{tabular}
[c]{ccccc|ccccc}
\hline
\multicolumn{5}{c}{} & \multicolumn{5}{|c}{MSPE}\\
\hline
Cases & $R_{gen}$ & $R$ & Node & Edge & \textbf{BNR} & Lasso &  Reli\'{o}n(2017) & BLasso & Horseshoe \\
      &           &     & Sparsity & Sparsity &      &       &                   &        & \\
\hline
Case - 1 & 3 & 5 & 0.7 & 0.5 & 0.119 & 0.151 & 0.425 & 0.125 &  \textbf{0.096}\\
Case - 2 & 3 & 5 & 0.2 & 0.5  & \textbf{0.451} & 0.549 & 0.699 &  0.692 & 0.566\\
\hline
\multicolumn{5}{c}{} & \multicolumn{5}{|c}{Coverage of 95\% PI}\\
\hline
Case - 1 & 3 & 5 & 0.7 & 0.5 & 0.93 & 1.00 & 0.86 & 0.96 &  0.96\\
Case - 2 & 3 & 5 & 0.2 & 0.5  & 1.00 & 0.83 & 0.70 & 1.00 &  1.00\\
\hline
\multicolumn{5}{c}{} & \multicolumn{5}{|c}{Length of 95\% PI}\\
\hline
Case - 1 & 3 & 5 & 0.7 & 0.5 & 6.18 & 14.19 & 8.35 & 6.44 &  5.91\\
Case - 2 & 3 & 5 & 0.2 & 0.5  & 41.69 & 27.98 & 17.84 &  51.70 & 49.12\\
\hline
\end{tabular}
\caption{MSPE, coverage and length of 95\% predictive intervals (PIs) of Bayesian Network Regression (BNR) vis-a-vis competitors for cases in \emph{Simulation 3}. Lowest MSPE for any case is made bold.}\label{Tab3prednew}
\end{center}
\end{table}

\section{Application to Human Brain Network Data}\label{sec5}
This section illustrates the inferential and predictive ability of Bayesian network regression in the context of a diffusion weighted magnetic resonance imaging (DWI) dataset. Along with the brain network data, the dataset of interest contains a measure of \emph{creativity} for several subjects, known as the Composite Creativity Index (CCI). The scientific goal in this setting pertains to understanding the relationship between brain connectivity and the composite creativity index (CCI). 
In particular, we are interested in predicting the CCI of a subject from his/her brain network, and to identify brain regions (nodes in the brain network) that are involved with creativity, as well as significant connections between different brain regions.

Human creativity has been at the crux of the evolution of the human civilization, and has been the topic of research in several disciplines including neuroscience. Though creativity can be defined in
numerous ways, one could envision a creative idea as one that is unusual as well as effective in a given social context (\citealp{flaherty2005frontotemporal}). Neuroscientists generally concur that a coalescence of several cognitive processes determines the creative process, which often involves a \emph{divergence of ideas} to conceivable solutions for a given problem. To measure the creativity of an individual, \citealp{jung2010neuroanatomy} propose the CCI, which is formulated
by linking measures of divergent thinking and creative achievement to cortical thickness of young (23.7 $\pm$ 4.2 years), healthy subjects. Three independent judges grade the creative products of a subject from which the ``composite creativity index" (CCI) is derived. CCI serves as the response in our study.

Along with CCI measurements, brain network information for $n = 37$ subjects is gathered using diffusion weighted magnetic resonance imaging (DWI). DWI  is an imaging technique that enables measurement of the restricted diffusion of water in tissue in order to produce neural tract images. The brain imaging data we use has been pre-processed using the \textsf{NDMG} pre-processing pipeline (\citealp{kiar2016ndmg}; \citealp{kiar2017example}; \citealp{kiar2017science}).
In the context of DWI, the human brain is divided according to the Desikan atlas (\citealp{desikan2006automated}) that identifies 34 cortical regions of interest (ROIs) both in the left and right hemispheres of the human brain, implying 68 cortical ROIs in all. 
A `brain network' for each subject is represented by a symmetric adjacency matrix whose rows and columns correspond to different ROIs and entries correspond to estimates of the number of `fibers' connecting pairs of brain regions. 
Thus, for each individual, representing the brain network, is a weighted adjacency matrix of dimension $68 \times 68$, with the $(k,l)$th off-diagonal entry in the adjacency matrix being the estimated number of fibers connecting the $k$th and the $l$th brain regions. Figure~\ref{DataPic} shows maps of the brain network for two representative individuals in the sample. Each cell of the adjacency matrix is standardized by subtracting the mean and dividing by the standard deviation with respect to all $n=37$ samples.  CCI is also standardized in a similar fashion. Now we fit our proposed model with standardized CCI as the response and the standardized adjacency matrix as the network predictor. We use identical prior distributions for all the parameters as in the simulation studies.

The MCMC chain is run for $50,000$ iterations, with the first $30,000$ iterations discarded as burn-in. Convergence is assessed by comparing different simulated sequences of representative parameters started at different initial values (\citealp{gelman2014bayesian}).
All inference is based on the remaining $20,000$ post burn-in iterates appropriately thinned. Moreover, we monitor the auto-correlation plots and effective sample sizes of the iterates.

\begin{figure}
  \begin{center}
    \subfigure[Representative Network Adjacency Matrix 1]{\includegraphics[page=2,width=8 cm,height=7cm]{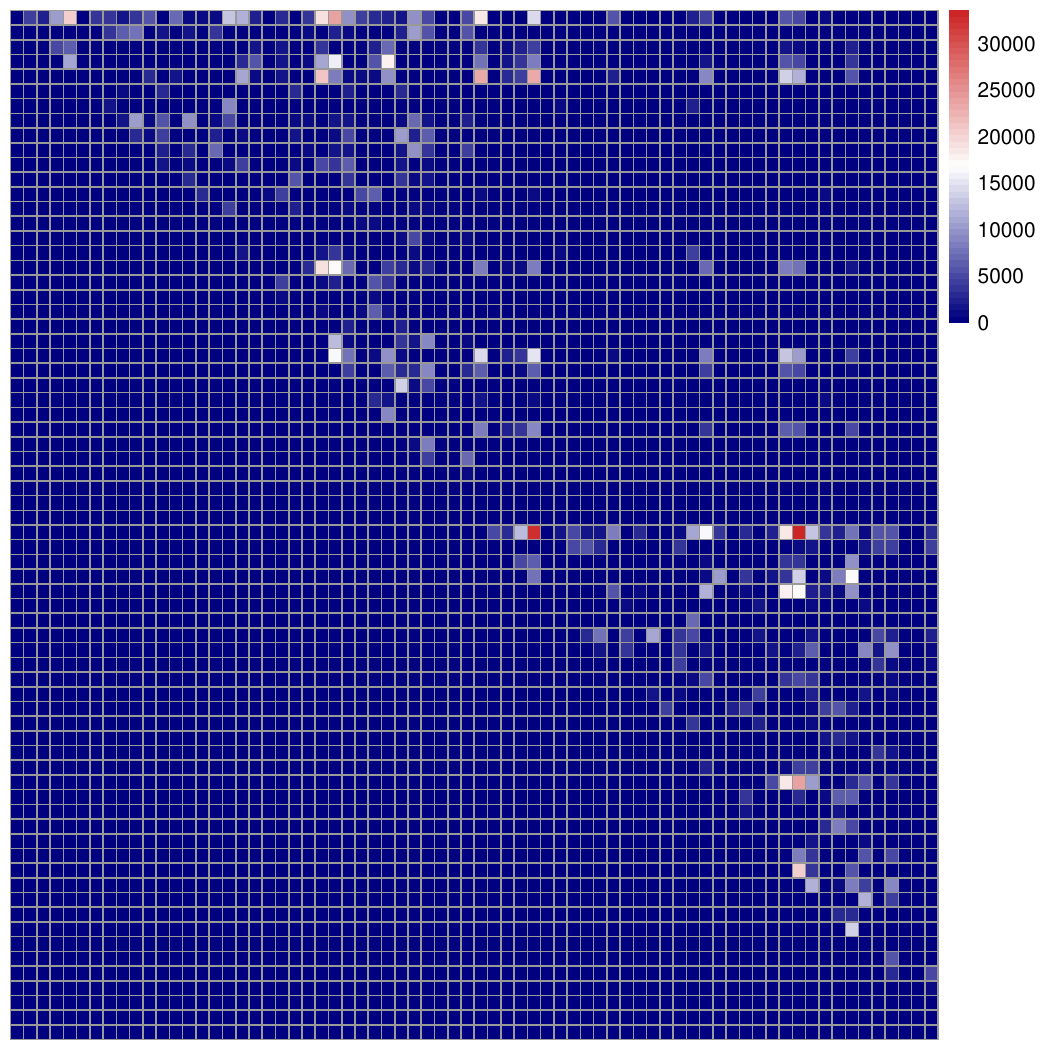}\label{Adj1}}
   \subfigure[Representative Network Adjacency Matrix 2]{\includegraphics[page=2,width=8 cm,height=7cm]{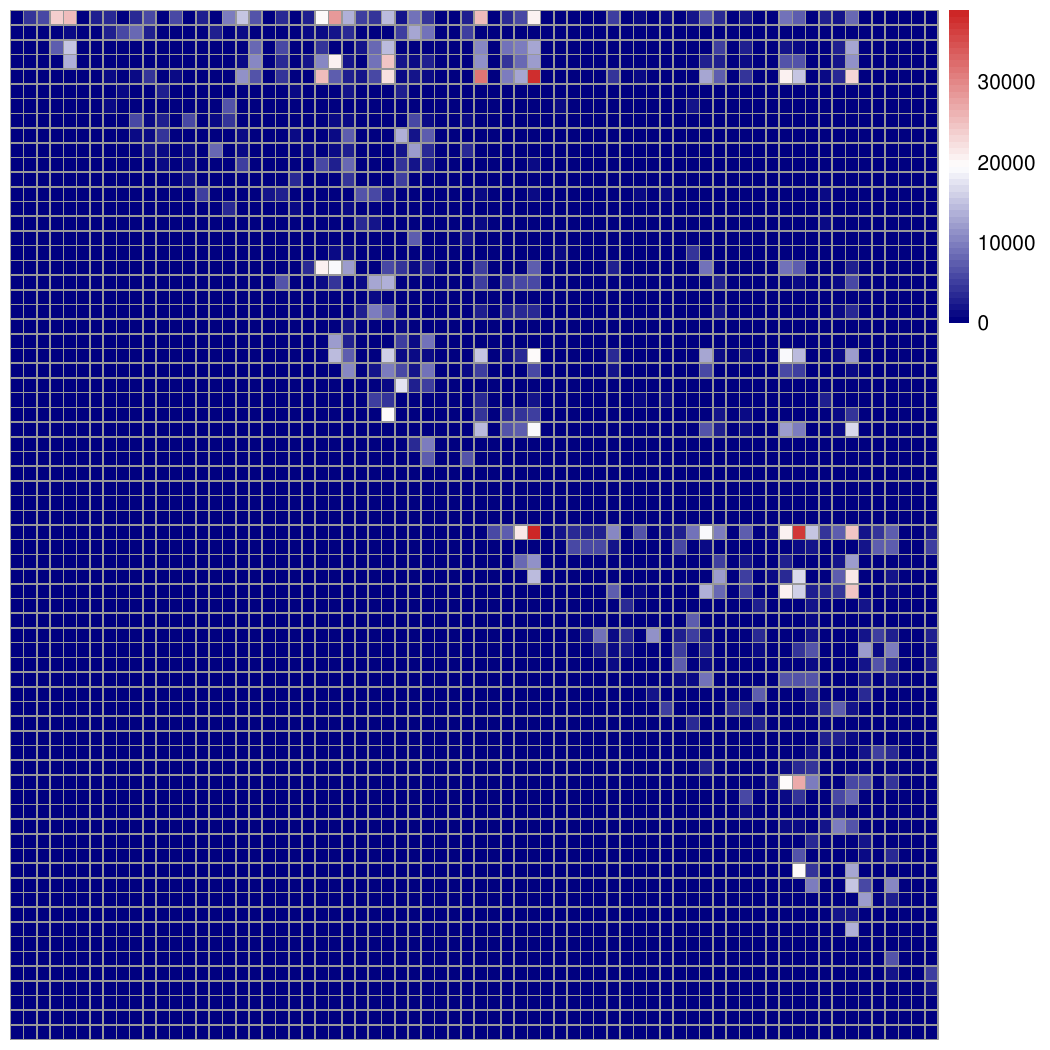}\label{Adj2}}
 \end{center}
 \caption{Figure shows maps of the brain network (weighted adjacency matrices) for two representative individuals in the sample. Since the $(k,l)$th off-diagonal entry in any adjacency matrix corresponds to the number of \emph{fibers} connecting the $k$th and the $l$th ROIs, the adjacency matrices are symmetric. Hence the figure only shows the upper triangular portion.}\label{DataPic}
\end{figure}

\subsection{Findings from BNR}
We focus on identifying influential ROIs in the brain network using the node selection strategy described in the simulation studies. For the purpose of this data analysis, the Bayesian network regression model is fitted with $R = 5$ which is found to be sufficient for this study. 
Recall that the $k$th node is identified as \emph{active} if $P(\xi_k = 1 \given Data)$ exceeds $0.5$. This criteria, when applied to the real data discussed above, identifies $36$ ROIs out of $68$ as \emph{active}. Of these $36$ ROIs, $16$ belong to the left portion of the brain (or the left hemisphere) and $20$ belong to the right hemisphere. The effective dimension of the model a-posteriori is $3$.
Table~\ref{Tab4} shows the brain ROIs in the Desikan atlas detected as being actively associated with the CCI. 

A large number of the $36$ active nodes detected by our method are part of the \emph{frontal} ($15$) and \emph{temporal} ($8$) cortices in both hemispheres. The frontal cortex has been scientifically associated with divergent thinking and problem solving ability, in addition to motor function, spontaneity, memory, language, initiation, judgement, impulse control and social behavior (\citealp{stuss1985subtle}). Some of the other functions directly related to the frontal cortex seem to be ``behavioral spontaneity,'' interpreting environmental feedback and risk taking (\citealp{razumnikova2007creativity}; \citealp{miller1985cognitive} ; \citealp{kolb1981performance}). On the other hand, \citealp{finkelstein1991impulsive} report \emph{de novo} artistic expression to be associated with the temporal and frontal regions.
Our method also finds a strong relationship between creativity and the \emph{right parahippocampal gyrus} and \emph{right inferior parietal lobule}, regions found to be involved with creativity by a few earlier scientific studies, see e.g., \citealp{chavez2004neurobiology}.

As a reference point to our analysis, we compare our findings with \citealp{jung2010neuroanatomy}, where a regression model is used to understand the relationship between CCI and ROI-specific measures to account for the relationship between creativity and different brain regions.
Our analysis finds a number of overlaps with the regions that \citealp{jung2010neuroanatomy} identify as significantly associated with the creativity process, namely the \emph{middle frontal gyrus}, the \emph{left cingulate cortex}, the \emph{left orbitofrontal} region, the \emph{left lingual} region, the \emph{right fusiform}, the \emph{left cuneus}, the \emph{right superior parietal lobule}, the \emph{inferior parietal}, the \emph{superior parietal} lobules and the \emph{right posterior singulate} regions. Although there is significant intersection between the findings of \citealp{jung2010neuroanatomy} and our method, there are a few regions that we detect as active and they do not, and vice versa. For example, our model detects the \emph{precuneus} and the \emph{supramarginal} regions in both the hemispheres to be significantly related to CCI, while \citealp{jung2010neuroanatomy} do not. On the other hand, they identify the \emph{right angular} region to be significant while we do not.
We also implement the method of \citealp{relion2017network} on our dataset, and find that it identifies $65$ out of $68$ ROIs as active. The regions that are found to inactive are the \emph{frontalpole}, \emph{temporalpole} and the \emph{transversetemporal} regions in the right hemisphere.

Along with influential ROIs, we are interested in identifying the statistically significant edges or connections between the $68$ ROIs. We consider the edge between two ROIs $k$ and $l$ to have a statistically significant impact on CCI if the $95\%$ credible interval of the posterior distribution of its corresponding coefficient $\gamma_{k,l}$ does not contain $0$. Under this measure, our model identifies $576$ significant $\gamma_{k,l}$'s. 
Figure \ref{Fig2} plots significant inter-connections detected among brain regions of interest (ROIs), where the brain can be viewed from different angles. Red dots show the active ROIs and blue lines show significant connections between them.
Figure \ref{Fig3} plots these influential interconnections, where a white cell represents an edge predictive of the response with the corresponding row ROI and column ROI. Since this is an undirected network, the matrix is symmetric and we only show connections in the upper triangular region.

Finally, our interest turns to the predictive ability of the Bayesian network regression model. To this end, Table~\ref{Tab3} reports the mean squared prediction error (MSPE) between observed and predicted responses, length and coverage of 95\% predictive intervals. Here, the average is computed over $10$ cross-validated folds.
As reference, we also present MSPE, length and coverage values for Lasso, BLasso and \citealp{relion2017network}. 
Our approach models correlation across coefficients sparsely, which seems to improve the prediction vis-a-vis Lasso and BLasso.
The table clearly shows excellent point prediction of our proposed approach even under small sample size and low signal to noise ratio.
Additionally, all competitors provide close to nominal coverage, with BNR yielding slightly less than nominal coverage, and the other competitors yielding slightly more, but with much wider credible intervals.

\begin{table}[!th]
\begin{center}
\begin{tabular}
[c]{cccccc}
\hline
 & BNR & Lasso &  BLasso & Reli\'{o}n(2017)\\
\hline
 MSPE & \textbf{0.69} & 0.98 & 1.84 & 0.98\\
 Coverage of 95\% PI & 0.93 & 0.97 & 0.97 & 0.97 \\
 Length of 95\% PI & \textbf{2.72} & 3.88 & 3.40 & 3.89\\
\hline
\end{tabular}
\caption{Predictive performance of competitors in terms of mean squared prediction error (MSPE), coverage and length of 95\% predictive intervals, obtained through $10$-Fold Cross Validation in the context of real data. Note that
since the response has been standardized, an MSPE value greater than or around $1$ will denote an inconsequential analysis. }\label{Tab3}
\end{center}
\end{table}

\begin{figure}[!ht]
   \begin{center}
   \includegraphics[width=11 cm, height = 8 cm]{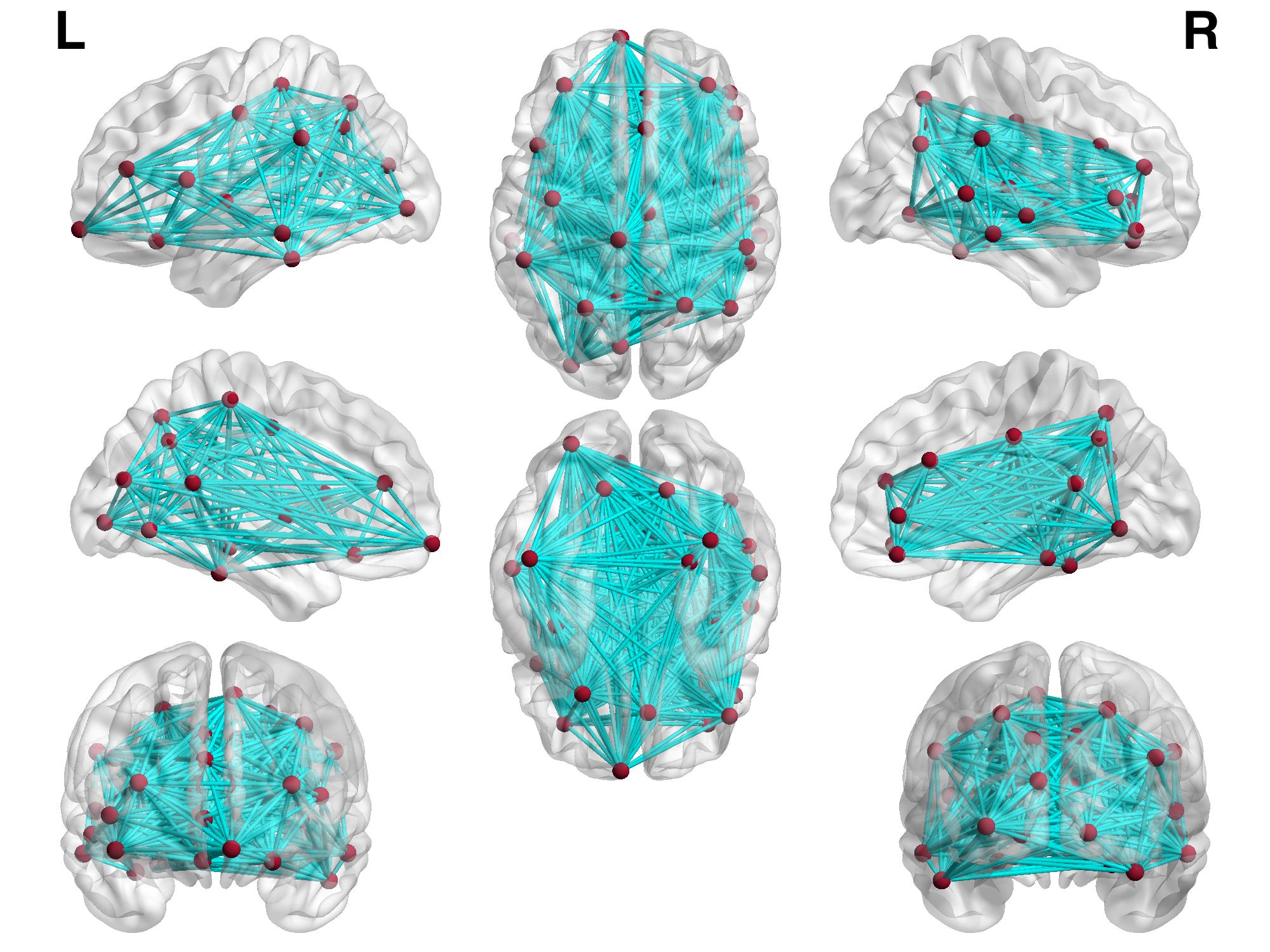}
   \end{center}
\caption{Significant inter-connections detected among brain regions of interest (ROIs) in the Desikan atlas. Red dots show the \emph{active} ROIs and blue lines show significant inter-connections between them.}
\label{Fig2}
\end{figure}

\begin{table}[h]
{\scriptsize
\centering
\begin{tabular}{cccccc}
\hline
\multicolumn{6}{|c|}{\textbf{Left Hemisphere Lobes}}\\
\cline{1-6}
         &                  &             &             &              &  \\
\textbf{Temporal} & \textbf{Cingulate} & \textbf{Frontal} & \textbf{Occipital} & \textbf{Parietal} & \textbf{Insula}\\
\hline
      \makecell{inferior \\temporal gyrus}        &     \makecell{isthmus \\ cingulate cortex}             &   \makecell{ lateral \\ orbitofrontal }        &      cuneus       &       precuneus      & insula  \\
      \hline
     \makecell{middle \\temporal gyrus} & \makecell{} & \makecell{paracentral} & \makecell{lateral \\occipital gyrus} & \makecell{superior \\parietal lobule}  &  \\
      \hline
       \makecell{} & \makecell{} & \makecell{pars \\ opercularis} & lingual & \makecell{supramarginal \\ gyrus} &  \\
        \hline
        & & precentral & & &\\
         \hline
        & & \makecell{rostral middle \\frontal gyrus} & & &\\
         \hline
        & & frontal pole & & &\\
        \hline

       &                  &             &             &              &  \\
        &                  &             &             &              &  \\
  \hline
  \multicolumn{6}{|c|}{\textbf{Right Hemisphere Lobes}}\\
  \cline{1-6}
   &                  &             &             &              &  \\
  \textbf{Temporal} & \textbf{Cingulate} & \textbf{Frontal} & \textbf{Occipital} & \textbf{Parietal} & \textbf{Insula}\\
  \hline
    \makecell{bank of the \\superior temporal sulcus}   &  \makecell{caudal \\ anterior cingulate}     &   \makecell{medial \\orbitofrontal}   &  lingual  &  \makecell{inferior \\parietal lobule}  & insula\\
      \hline
     \makecell{fusiform} & \makecell{isthmus \\cingulate cortex} & \makecell{pars \\orbitalis}  &  &  \makecell{precuneus} &  \\
      \hline
       \makecell{middle\\temporal gyrus} & \makecell{posterior\\cingulate cortex} & \makecell{pars \\ triangularis} &  & \makecell{superior \\parietal lobule} &  \\
        \hline
         parahippocampal &  \makecell{rostral \\anterior cingulate cortex}  &  \makecell{rostral \\middle frontal gyrus} & & \makecell{supramarginal \\gyrus} &\\
         \hline
          \makecell{superior \\ temporal gyrus}  & &  & & &\\
         \hline
         \makecell{transverse \\ temporal}  & & & & &\\
         \hline
        & & & & &\\

\end{tabular}
\caption{Brain regions (ROIs) detected as actively associated with the composite creativity index by BNR. }\label{Tab4}
}
\end{table}

\begin{figure}[!ht]
   \begin{center}
   \includegraphics[width=15.5 cm, height = 15.5 cm]{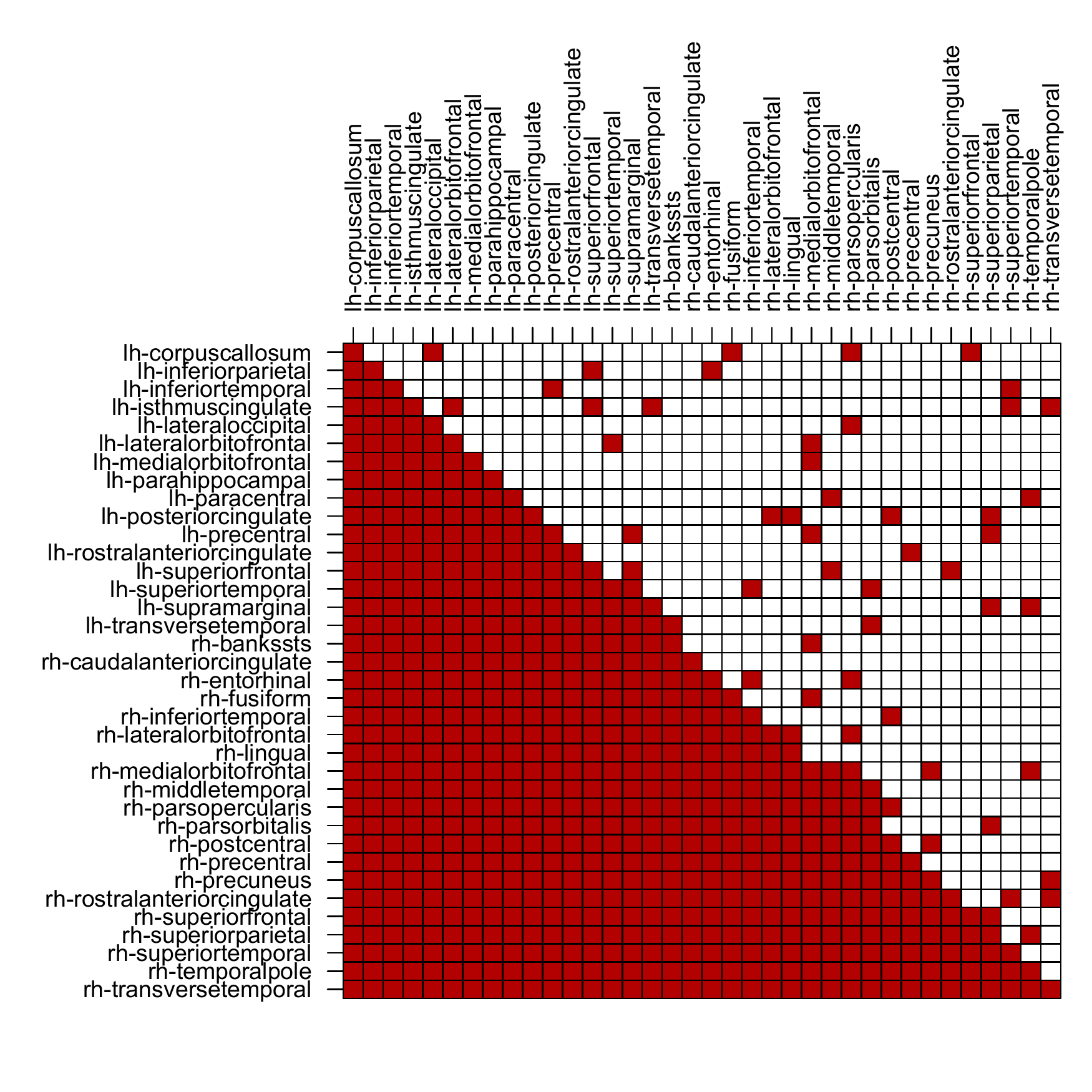}
   \end{center}
\caption{Significant inter-connections detected among \emph{active} brain regions of interest (ROIs) in the Desikan atlas. White cells show significant nodal associations among ROIs denoted by the corresponding rows and columns. Interconnections have been shown only in the upper triangular portion. Prefix `lh-' and `rh-' in the ROI names denote their positions in the left and right hemispheres of the brain respectively. For full names of the ROIs specified on the axes, please consult the widely available Desikan Atlas.}
\label{Fig3}
\end{figure}

\section{Conclusion and Future Work}\label{sec6}
This article proposes a novel Bayesian framework to address a regression problem with a continuous response and network-valued predictors, respecting the underlying network structure.
Our contribution lies in carefully constructing a novel class of network shrinkage priors corresponding to the network predictor which accounts for the correlation in the regression coefficients that is expected from the relational nature of the predictor.  Empirical results from simulation studies show that our method is superior to popular alternatives in situations where the regression coefficients have a network structure, and very competitive in other circumstances, both in terms of inference as well as prediction. Our framework is employed to analyze
a brain connectome dataset on composite creativity index along with the brain network of multiple individuals. It is able to identify important regions in the brain and important brain connectivity patterns which have profound influence on the creativity of a person.

A number of future directions emerge from this work. First, our framework finds natural extension to regression problems with a binary response and \emph{any} network predictor,
whether binary or weighted. Such a framework would be useful in various classification problems involving network predictors, e.g., in classifying diseased patients from normal people in neuroimaging studies. Another important
direction appears to be the development of a regression framework with the network as the response regressed on a few scalar/vector predictors. Some of these constitute our current work.

\section*{Appendix A}\label{appA}
This section provides details of posterior computation for all the parameters in the Bayesian network regression with a continuous response.

Let $\bx_i=(a_{i,1,2},a_{i,1,3},...,a_{i,1,V}, a_{i,2,3}, a_{i,2,4},..., a_{i,2,V}, ...., a_{i,V-1,V})'$ be of dimension
$q\times 1$, where $q=\frac{V\times(V-1)}{2}$. Assume $\by=(y_1,...,y_n)'\in\mathbb{R}^n$ and $\bX=(\bx_1:\cdots:\bx_n)'$ is an $n\times q$ matrix. Further, assume $\bW=(\bu_1'\bLambda\bu_2,...,\bu_1'\bLambda\bu_V,....,\bu_{V-1}'\bLambda\bu_V)'$, $\bD=diag(s_{1,2},...,s_{V-1,V})$ and $\bgamma=(\gamma_{1,2},...,\gamma_{V-1,V})'$.
Thus, with $n$ data points, the hierarchical model with the Bayesian Network Lasso prior can be written as
\begin{align*}
&\qquad\qquad\qquad\qquad\qquad\by \sim \mathrm{N}(\mu+\bX\bgamma,\tau^2\bI)\\
&\bgamma \sim \mathrm{N}(\bW,\tau^2\bD),\:\:
(\mu,\tau^2) \sim \pi(\mu,\tau^2) \propto \frac{1}{\tau^2},\:\:
\bu_k|\xi_k=1 \sim N(\bu_k \given \bzero, \bM),\:\bu_k|\xi_k=0\sim\delta_{\bzero},\:\mu\sim flat() \\
& s_{k,l} \sim Exp(\theta^2/2),\:\:\:\:\theta^2 \sim Gamma(\zeta,\iota),\:\:
\bM \sim IW(\bS,\nu),\:\:
\Delta \sim Beta(a_{\Delta},b_{\Delta}),\xi_k\sim Ber(\Delta)\\
&\qquad\qquad\qquad \lambda_r \sim Ber (\pi_{r}), \:\:
 \pi_{r} \sim Beta(1, r^{\eta}),\:\eta>1.
\end{align*}

The hierarchical model specified above leads to straightforward Gibbs sampling with full conditionals obtained as following:
\begin{itemize}
\item $\mu \given - \sim N\left(\frac{{\boldsymbol 1}'(\by-\bX\bgamma)}{n},\frac{\tau^2}{n}\right)$
\item $\bgamma \given - \sim N(\bmu_{\bgamma \given \cdot}, \bSigma_{\bgamma \given \cdot})$,
where $\bmu_{\bgamma \given \cdot} = {(\bX'\bX + \bD^{-1})}^{-1}(\bX'(\by-\mu{\boldsymbol 1}) + \bD^{-1}\bW)$ and $\bSigma_{\bgamma \given \cdot} = \tau^2 {(\bX'\bX + \bD^{-1})}^{-1}$
\item
$\tau^2 \given - \sim IG\left[(\frac{n}{2} + \frac{V(V-1)}{4}), \frac{(\by-\mu{\boldsymbol 1}-\bX\bgamma)'(\by-\bX\bgamma) + (\bgamma - \bW)'\bD^{-1}(\bgamma - \bW)}{2} \right]$

\item $s_{k,l} \given - \sim GIG\left[\frac{1}{2}, \frac{(\gamma_{k,l} - \bu_k '\bLambda \bu_l)^2}{\tau^2}, \theta^2 \right]$, where GIG denotes the generalized inverse Gaussian distribution.
\item $\theta^2 \given - \sim Gamma\left[ \left({\zeta} + \frac{V(V-1)}{2}\right), \left(\iota + \sum_{k < l} \frac{s_{k,l}}{2} \right) \right]$
\item $\bu_k \given - \sim   w_{\bu_k} \: \delta_0 (\bu_k)  +  (1 - w_{\bu_k}) \: N(\bu_k \given \bm_{\bu_k}, \bSigma_{\bu_k})$, where
$\bU^{\ast}_k=(\bu_1:\cdots:\bu_{k-1}:\bu_{k+1}:\cdots:\bu_{V})' \bLambda,\:\:\bH_k=diag(s_{1,k},...,s_{k-1,k},s_{k,k+1},...,s_{k,V}),\:\:\bgamma_k=(\gamma_{1,k},...,\gamma_{k-1,k},\gamma_{k,k+1},...,\gamma_{k,V})$, and
\begin{align*}
&\bSigma_{\bu_k} = \left(\bU^{\ast'}_h\bH_k^{-1}\bU^{\ast}_k/\tau^2+\bM^{-1}\right)^{-1},\:\:\bm_{\bu_k}=\bSigma_{\bu_k}\bU^{\ast'}_k\bH_k^{-1}\bgamma_k/\tau^2\\
& w_{\bu_k} = \frac{(1-\pi)N(\bgamma_k\given \bzero,\tau^2\bH_k)}{(1-\pi)N(\bgamma_k\given \bzero,\tau^2\bH_k)+\pi N(\bgamma_k\given \bzero,\tau^2\bH_k+\bU^{\ast}_k\bM\bU^{\ast'}_k)}
\end{align*}
\item $\xi_k|-\sim Ber(1 - w_{\bu_k})$
\item $\Delta \given - \sim Beta\left[(a_{\Delta} + \sum_{k = 1}^{V} \xi_k),   (b_{\Delta} + \sum_{k = 1}^{V} (1 - \xi_k))\right]$.
\item $\bM \given - \sim IW [(\bS + \sum_{k : \bu_k \neq \bzero} \bu_k\bLambda \bu_k'),(\nu + \{\#k : \bu_k \neq \bzero \}) ]$.
\item $\lambda_r \given - \sim Ber(p_{\lambda_r})$, where
       $p_{\lambda_r}=\frac{\pi_{r}N(\bgamma\given \bW_1,\tau^2\bD)}{\pi_{r}N(\bgamma\given \bW_1,\tau^2\bD)+
       (1-\pi_{r})N(\bgamma\given \bW_0,\tau^2\bD)}$. Here \\ $\bW_1=(\bu_1'\bLambda_1\bu_2,...,\bu_1'\bLambda_1\bu_V,....,\bu_{V-1}'\bLambda_1\bu_V)'$, $\bW_0=(\bu_1'\bLambda_0\bu_2,...,\bu_1'\bLambda_0\bu_V,....,\bu_{V-1}'\bLambda_0\bu_V)'$,
       $\bLambda_1=diag(\lambda_1,..,\lambda_{r-1},1,\lambda_{r+1},..,\lambda_R)$, $\bLambda_0=diag(\lambda_1,..,\lambda_{r-1},0,\lambda_{r+1},..,\lambda_R)$, for $r=1,..,R$.
\item $\pi_{r}\given - \sim Beta(\lambda_r+1,1-\lambda_r+r^{\eta})$, for $r=1,..,R$.
\end{itemize}

\section*{Appendix B}\label{appB}
This section shows the posterior propriety of the parameters in the BNR model. Without loss of generality, we set $\mu=0$ while proving the posterior propriety. To begin with, we state a number of useful lemmas.
\subsection*{Preliminary Results}
\begin{lemma}\label{lem1}
If $\bC$ is an $h\times h$ non-negative definite matrix, then $|\bC+\bI|\geq 1$.
\end{lemma}
\begin{proof}
The eigenvalues of $(\bC+\bI)$ are given by $\varphi_1+1,...,\varphi_h+1$, where $\varphi_1,...,\varphi_h$ are eigenvalues of
$\bC$. Since $\bC$ is non-negative definite, $\varphi_1\geq 0,...,\varphi_h\geq 0$. The result follows from the fact that
$|\bC+\bI|=\prod_{l=1}^{h}(\varphi_l+1)$ is the product of eigenvalues.
\end{proof}

\begin{lemma}\label{lem2}
Let $\bC$ be an $h\times h$ diagonal matrix with diagonal entries $c_1,...,c_h$ all greater than $0$. Suppose $\bA$ is an $n\times h$ matrix
with the largest eigenvalue of $\bA\bA'$ given by $\mu_{\bA\bA'}$. Then
$\bA\bC\bA'+\bI\leq \left(\mu_{\bA\bA'}\sum_{l=1}^{h}c_l+1\right)\bI$, where $\bH_1\leq \bH_2$ implies $\bH_2-\bH_1$ is a positive definite
matrix.
\end{lemma}
\begin{proof}
Since $c_1,...,c_h>0$, $\bA\bC\bA'\leq (\sum_{l=1}^{h}c_l)\bA\bA'$. Consider the spectral decomposition of the matrix $\bA\bA'$. Let
the eigen-decomposition of $\bA\bA'=\bLambda\bH\bLambda'$, where $\bLambda$ is the matrix of eigenvectors and $\bH$ is a diagonal matrix
with diagonal entries $\mu_1,...,\mu_n$. Since each $\mu_i\leq \mu_{\bA\bA'}$,
$\bA\bA'\leq\mu_{\bA\bA'}\bLambda\bLambda'=\mu_{\bA\bA'}\bI$. Thus,
$\bA\bC\bA'\leq (\sum_{l=1}^{h}c_l)\mu_{\bA\bA'}\bI$. Hence $\bA\bC\bA'+\bI\leq \left(\mu_{\bA\bA'}\sum_{l=1}^{h}c_l+1\right)\bI$.
\end{proof}

\begin{lemma}\label{lem3}
Suppose $\bz$ is an $h\times 1$ vector and $\bA$ is an $h\times h$ symmetric positive definite matrix. Let $\bB$ be another $h\times h$ positive
definite matrix such that $\bA\geq\bB$ (where $\bA\geq\bB$ implies $\bA-\bB$ is non-negative definite). Then $\bz'\bA^{-1}\bz\leq \bz'\bB^{-1}\bz$.
\end{lemma}
\begin{proof}
$\bA\geq\bB$ implies $\bB^{-1/2}\bA\bB^{-1/2}\geq\bI$. Thus all eigenvalues of $\bB^{-1/2}\bA\bB^{-1/2}=\bB^{-1/2}\bA^{1/2}\bA^{1/2}\bB^{-1/2}$ are greater than or equal to 1. Since commuting the product of two matrices does not change the eigenvalues, $\bA^{1/2}\bB^{-1}\bA^{1/2}$ has all eigenvalues greater than or equal to 1. Thus $\bA^{1/2}\bB^{-1}\bA^{1/2}\geq\bI$, which implies  $\bA^{-1}\leq\bB^{-1}$. Then $\bz'\bA^{-1}\bz\leq \bz'\bB^{-1}\bz$.
\end{proof}

\subsection*{Main Result}
Note that the posterior distribution of the parameters is given by
\begin{multline*}
p(\bgamma,\tau^2,\bu_1,..,\bu_V,\xi_1,..,\xi_V,\lambda_1,..,\lambda_R,\theta^2,\Delta,\{s_{k,l}\}_{k<l},\pi_1,...,\pi_R,\bM\given \by,\bX)\\
\propto
\mathrm{N}(\by\given\bX\bgamma,\tau^2\bI)\times\mathrm{N}(\bgamma\given\bW,\tau^2\bD)\times
\frac{1}{\tau^2}\times
\prod\limits_{k=1}^{V}\left[\xi_k N(\bu_k \given \bzero, \bM)+(1-\xi_k)\delta_{\bzero}\right]\\
\qquad\times\prod\limits_{k<l} Exp(s_{k,l}\given\theta^2/2)\times
Gamma(\theta^2 \given \zeta,\iota)\times IW(\bM \given\bS,\nu)\times Beta(\Delta \given a_{\Delta},b_{\Delta})\\
\qquad\times \prod\limits_{r=1}^{R}\left[Ber (\lambda_r \given \pi_{r}) \times Beta(\pi_{r}\given 1, r^{\eta})\right]\times \prod_{k=1}^{V}Ber(\xi_k\given\Delta).
\end{multline*}
Integrating over $\xi_1,...,\xi_V$
\begin{multline*}
p(\bgamma,\tau^2,\bu_1,..,\bu_V,\lambda_1,..,\lambda_R,\theta^2,\Delta,\{s_{k,l}\}_{k<l},\pi_1,...,\pi_R,\bM\given \by,\bX)\propto
\mathrm{N}(\by\given\bX\bgamma,\tau^2\bI)\times\\
\qquad \mathrm{N}(\bgamma\given\bW,\tau^2\bD)\times\frac{1}{\tau^2}\times
\prod\limits_{k=1}^{V}\left[\Delta N(\bu_k \given \bzero, \bM)+(1-\Delta)\delta_{\bzero}\right]\times\prod\limits_{k<l} Exp(s_{k,l}\given\theta^2/2)\times \\
\qquad Gamma(\theta^2 \given \zeta,\iota)\times IW(\bM \given\bS,\nu)\times Beta(\Delta \given a_{\Delta},b_{\Delta})\times
\prod\limits_{r=1}^{R}\left[Ber(\lambda_r \given \pi_{r}) \times  Beta(\pi_{r}\given 1, r^{\eta})\right].
\end{multline*}
Further integrating over $\pi_1,...,\pi_R$ yields,
\begin{multline*}
p(\bgamma,\tau^2,\bu_1,..,\bu_V,\lambda_1,..,\lambda_R,\theta^2,\Delta,\{s_{k,l}\}_{k<l},\bM\given \by,\bX)\propto
\mathrm{N}(\by\given\bX\bgamma,\tau^2\bI)\times \mathrm{N}(\bgamma\given\bW,\tau^2\bD)\times\\
\qquad \frac{1}{\tau^2}\times
\prod\limits_{k=1}^{V}\left[\Delta N(\bu_k \given \bzero, \bM)+(1-\Delta)\delta_{\bzero}\right]\times\prod\limits_{k<l} Exp(s_{k,l}\given\theta^2/2)\times Gamma(\theta^2 \given \zeta,\iota)\times\\
\qquad IW(\bM \given\bS,\nu)\times Beta(\Delta \given a_{\Delta},b_{\Delta})\times
\prod\limits_{r=1}^{R}\frac{\Gamma(\lambda_r+1)\Gamma(1-\lambda_r+r^{\eta})\Gamma(r^{\eta}+1)}{\Gamma(r^{\eta}+2)\Gamma(r^{\eta})}.
\end{multline*}

The prior specifications on $\Delta$ enable it to be bounded within a finite interval of (0,1). Thus in showing the posterior propriety of parameters with unbounded range, it is enough to treat $\Delta$ as constant. We treat it as fixed henceforth.

Note that each $\lambda_r\in\{0,1\}$, hence marginalizing out $\lambda_r$ gives
\begin{multline*}
p(\bgamma,\bLambda,\tau^2,\bu_1,..,\bu_V,\theta^2,\{s_{k,l}\}_{k<l},\bM\given \by,\bX)\propto
\sum\limits_{\lambda_r\in\{0,1\}}\Big[\mathrm{N}(\by\given\bX\bgamma,\tau^2\bI)\times \mathrm{N}(\bgamma\given\bW,\tau^2\bD)\times\\
\qquad \frac{1}{\tau^2}\times
\prod\limits_{k=1}^{V}\left[\Delta N(\bu_k \given \bzero, \bM)+(1-\Delta)\delta_{\bzero}\right]\times\prod\limits_{k<l} Exp(s_{k,l}\given\theta^2/2)\times Gamma(\theta^2 \given \zeta,\iota)\times\\
\qquad IW(\bM \given\bS,\nu)\times
\prod\limits_{r=1}^{R}\frac{\Gamma(\lambda_r+1)\Gamma(1-\lambda_r+r^{\eta})\Gamma(r^{\eta}+1)}{\Gamma(r^{\eta}+2)\Gamma(r^{\eta})}\Big].
\end{multline*}

Integrating over $\bgamma$, we obtain,
\begin{multline*}
p(\bu_1,..,\bu_V,\tau^2,\theta^2,\{s_{k,l}\}_{k<l},\bM\given \by,\bX)\propto \sum\limits_{\lambda_r\in\{0,1\}}\Big[\frac{1}{(\tau^2)^{n/2+1}|\bX\bD\bX'+\bI|^{1/2}}\times\\
 \qquad \exp\left\{-\frac{(\by-\bX\bW)'(\bX\bD\bX'+\bI)^{-1}(\by-\bX\bW)}{2\tau^2}\right\}\times
\prod\limits_{k=1}^{V}\left[\Delta N(\bu_k \given \bzero, \bM)+(1-\Delta)\delta_{\bzero}\right]\times \\
\qquad\prod\limits_{k<l} Exp(s_{k,l}\given\theta^2/2)\times Gamma(\theta^2 \given \zeta,\iota)\times IW(\bM \given\bS,\nu)\times\\
\prod\limits_{r=1}^{R}\frac{\Gamma(\lambda_r+1)\Gamma(1-\lambda_r+r^{\eta})\Gamma(r^{\eta}+1)}{\Gamma(r^{\eta}+2)\Gamma(r^{\eta})}\Big].
\end{multline*}
Next, we integrate w.r.t. $\theta^2$ to obtain
\begin{align}\label{eq:posterior_prop1}
& p(\bu_1,..,\bu_V,\tau^2,\{s_{k,l}\}_{k<l},\bM\given \by,\bX)\propto \sum\limits_{\lambda_r\in\{0,1\}}\Big[\frac{1}{(\tau^2)^{n/2+1}|\bX\bD\bX'+\bI|^{1/2}}\times\nonumber\\
& \qquad \exp\left\{-\frac{(\by-\bX\bW)'(\bX\bD\bX'+\bI)^{-1}(\by-\bX\bW)}{2\tau^2}\right\}\times
\prod\limits_{k=1}^{V}\left[\Delta N(\bu_k \given \bzero, \bM)+(1-\Delta)\delta_{\bzero}\right]\times\nonumber \\
&\qquad\frac{1}{(\iota+\sum\limits_{k<l}s_{k,l})^{q+\zeta}}\times IW(\bM \given\bS,\nu)\times
\prod\limits_{r=1}^{R}\frac{\Gamma(\lambda_r+1)\Gamma(1-\lambda_r+r^{\eta})\Gamma(r^{\eta}+1)}{\Gamma(r^{\eta}+2)\Gamma(r^{\eta})}\Big].
\end{align}
(\ref{eq:posterior_prop1}) is a discrete sum of $2^R$ terms with different combinations of $\lambda_1,...,\lambda_r$. The sum integrated out over all the parameters is finite if the individual summands are finite when integrated out w.r.t all parameters. \\
Denote a representative summand by $p^{*}(\bu_1,..,\bu_V,\tau^2,\{s_{k,l}\}_{k<l},\bM\given \by,\bX)$, where
\begin{multline}
p^{*}(\bu_1,..,\bu_V,\tau^2,\{s_{k,l}\}_{k<l},\bM\given \by,\bX)\propto
\frac{1}{(\tau^2)^{n/2+1}|\bX\bD\bX'+\bI|^{1/2}}\times\nonumber\\
\qquad \exp\left\{-\frac{(\by-\bX\bW)'(\bX\bD\bX'+\bI)^{-1}(\by-\bX\bW)}{2\tau^2}\right\}\times
\prod\limits_{k=1}^{V}\left[\Delta N(\bu_k \given \bzero, \bM)+(1-\Delta)\delta_{\bzero}\right]\times\nonumber \\
\qquad\frac{1}{(\iota+\sum\limits_{k<l}s_{k,l})^{q+\zeta}}\times IW(\bM \given\bS,\nu)\times
\prod\limits_{r=1}^{R}\frac{\Gamma(\lambda_r+1)\Gamma(1-\lambda_r+r^{\eta})\Gamma(r^{\eta}+1)}{\Gamma(r^{\eta}+2)\Gamma(r^{\eta})}.
\end{multline}

Note the fact that $\bD$ is a diagonal matrix with all positive diagonal entries. Thus $\bX\bD\bX'$ is non-negative definite and by using  Lemma~\ref{lem2}
\begin{align*}
\bX\bD\bX'+\bI\leq \bX\bX'\sum_{k<l}s_{k,l}+\bI\leq \left(\mu_{\bX\bX'}\sum_{k<l}s_{k,l}+1\right)\bI,
\end{align*}
where $\bA\leq \bB$ implies $\bA-\bB$ is a non-negative definite matrix and $\mu_{\bX\bX'}$ is the largest eigenvalue of $\bX\bX'$.
Using Lemma~\ref{lem3}, the above inequality implies
\begin{align*}
(\by-\bX\bW)'(\bX\bD\bX'+\bI)^{-1}(\by-\bX\bW)\geq \frac{||\by-\bX\bW||^2}{\mu_{\bX\bX'}\sum_{k<l}s_{k,l}+1}.
\end{align*}
Let
\begin{align}\label{eq:posterior_prop2}
& \tilde{p}(\bu_1,..,\bu_V,\tau^2,\{s_{k,l}\}_{k<l},\bM)= \frac{1}{(\tau^2)^{n/2+1}|\bX\bD\bX'+\bI|^{1/2}}\times\nonumber\\
& \qquad \exp\left\{-\frac{(\by-\bX\bW)'(\bX\bD\bX'+\bI)^{-1}(\by-\bX\bW)}{2\tau^2}\right\}\times
\prod\limits_{k=1}^{V}N(\bu_k \given \bzero, \bM)\times\nonumber \\
&\qquad\frac{1}{(\iota+\sum\limits_{k<l}s_{k,l})^{q+\zeta}}\times IW(\bM \given\bS,\nu)\times
\prod\limits_{r=1}^{R}\frac{\Gamma(\lambda_r+1)\Gamma(1-\lambda_r+r^{\eta})\Gamma(r^{\eta}+1)}{\Gamma(r^{\eta}+2)\Gamma(r^{\eta})}.
\end{align}
With little algebra it can be shown that
\begin{align*}
& p^{*}(\bu_1,..,\bu_V,\tau^2,\{s_{k,l}\}_{k<l},\bM\given \by,\bX)\\
&=\mbox{constant}\times\sum\limits_{1\leq j_1,...,j_l\leq V, 0\leq l\leq V}\Delta^l(1-\Delta)^{V-l}
\tilde{p}(\bu_{j_1},..,\bu_{j_l},\bu_{j_{l+1}}=0,..,\bu_{j_V}=0,\tau^2,\{s_{k,l}\}_{k<l},\bM).
\end{align*}
Therefore, the integral of (\ref{eq:posterior_prop1}) w.r.t. all parameters is finite if and only if 
\begin{align*}
\int \tilde{p}(\bu_1,..,\bu_V,\tau^2,\{s_{k,l}\}_{k<l},\bM) d\bu_1\cdots d\bu_Vd\tau^2 d\bM d\prod\limits_{k<l}s_{k,l}  <\infty.
\end{align*}
Henceforth, we will proceed to show that this integral is finite.

With little algebra, we have that 
\begin{align*}
\int IW(\bM|\nu,\bS)\prod\limits_{k=1}^{V}N(\bu_k \given \bzero, \bM)d\bM\propto\frac{1}{|\bS+\sum_{k=1}^{V}\bu_k\bu_k'|^{(\nu+V)/2}}. 
\end{align*}
Hence,
\begin{multline*}
\tilde{p}(\bu_1,..,\bu_V,\tau^2,\{s_{k,l}\}_{k<l})\leq \mbox{constant}\times\frac{1}{|\bS+\sum_{k=1}^{V}\bu_k\bu_k'|^{(\nu+V)/2}}\frac{1}{(\tau^2)^{n/2+1}}\times\nonumber\\
\qquad \exp\left\{-\frac{||\by-\bX\bW||^2}{2\tau^2(\mu_{\bX\bX'}\sum_{k<l}s_{k,l}+1)}\right\}\times
\frac{1}{(\iota+\sum\limits_{k<l}s_{k,l})^{q+\zeta}}\frac{1}{|\bX\bD\bX'+\bI|^{1/2}}\times\\
\qquad \prod\limits_{r=1}^{R}\frac{\Gamma(\lambda_r+1)\Gamma(1-\lambda_r+r^{\eta})\Gamma(r^{\eta}+1)}{\Gamma(r^{\eta}+2)\Gamma(r^{\eta})}.
\end{multline*}
Define $\mathcal{A}=\left\{(\bu_1,...,\bu_V):||\by-\bX\bW||^2>1\right\}$. Then
\begin{multline}\label{eq:posterior_prop3}
\int \tilde{p}(\bu_1,..,\bu_V,\tau^2,\{s_{k,l}\}_{k<l})d\bu_1\cdots d\bu_V d\tau^2 d\prod\limits_{k<l}s_{k,l} \\
=\int\limits_{\mathcal{A}}\tilde{p}(\bu_1,..,\bu_V,\tau^2,\{s_{k,l}\}_{k<l}) d\bu_1\cdots d\bu_Vd\tau^2 d\prod\limits_{k<l}s_{k,l}+\nonumber\\
\int\limits_{\mathcal{A}^c}\tilde{p}(\bu_1,..,\bu_V,\tau^2,\{s_{k,l}\}_{k<l})d\bu_1\cdots d\bu_Vd\tau^2 d\prod\limits_{k<l}s_{k,l}.
\end{multline}
Now,
\begin{multline*}
 \int\limits_{\mathcal{A}}\tilde{p}(\bu_1,..,\bu_V,\tau^2,\{s_{k,l}\}_{k<l})d\tau^2 d\prod\limits_{k<l}s_{k,l}d\bu_1\cdots d\bu_V \leq \mbox{constant}\int\limits_{\mathcal{A}}\frac{d\bu_1\cdots d\bu_V}{|\bS+\sum_{k=1}^{V}\bu_k\bu_k'|^{(\nu+V)/2}}\times\\
\qquad \int\frac{1}{(\tau^2)^{n/2+1}}\exp\left\{-\frac{1}{2\tau^2(\mu_{\bX\bX'}\sum_{k<l}s_{k,l}+1)}\right\}\times
\frac{1}{(\iota+\sum\limits_{k<l}s_{k,l})^{q+\zeta}}\frac{d\tau^2 d\prod\limits_{k<l}s_{k,l}}{|\bX\bD\bX'+\bI|^{1/2}}\times\\
\qquad \prod\limits_{r=1}^{R}\frac{\Gamma(\lambda_r+1)\Gamma(1-\lambda_r+r^{\eta})\Gamma(r^{\eta}+1)}{\Gamma(r^{\eta}+2)\Gamma(r^{\eta})}\\
\leq \mbox{constant}\left\{\int\limits_{\mathcal{A}} \frac{1}{|\bS+\sum_{k=1}^{V}\bu_k\bu_k'|^{(\nu+V)/2}}d\bu_1\cdots d\bu_V\right\}\times\\
\left\{\int \frac{(\mu_{\bX\bX'}\sum_{k<l}s_{k,l}+1)^{n/2}}{|\bX\bD\bX'+\bI|^{1/2}(\iota+\sum\limits_{k<l}s_{k,l})^{q+\zeta}} d\prod_{k<l}s_{k,l}\right\}\times
\prod\limits_{r=1}^{R}\frac{\Gamma(\lambda_r+1)\Gamma(1-\lambda_r+r^{\eta})\Gamma(r^{\eta}+1)}{\Gamma(r^{\eta}+2)\Gamma(r^{\eta})}.
\end{multline*}
Note that
\begin{multline*}
\int\limits_{\mathcal{A}} \frac{1}{|\bS+\sum_{k=1}^{V}\bu_k\bu_k'|^{(\nu+V)/2}}d\bu_1\cdots d\bu_V\leq \int\limits_{\mathcal{A}} \frac{1}{\prod_{k=1}^{V}|\bS+\bu_k\bu_k'|^{\nu/2V+1/2}}d\bu_1\cdots d\bu_V\\
\leq \prod_{k=1}^{V}\left(\int\limits_{\mathcal{A}} \frac{1}{|\bS+\bu_k\bu_k'|^{\nu/V+1}}d\bu_k\right)^{1/2},
\end{multline*}
where the first inequality follows from the fact that $|\bS+\sum_{k=1}^{V}\bu_k\bu_k'|\geq |\bS+\bu_k\bu_k'|$ for all $k$. The second inequality is a direct application of the Cauchy-Schwarz inequality. By the ratio test of integrals, this integral is finite if $\int \frac{1}{[(1+u_{k,1})^2\cdots(1+u_{k,R})^2]^{\nu/V+1}}d\bu_k$ is finite. Now use the
 fact that $\int \frac{1}{x^{1+c}}dx<\infty$ for any $c>0$ to argue that $\int \frac{1}{[(1+u_{k,1})^2\cdots(1+u_{k,R})^2]^{2\nu/V+1}}d\bu_k$ is finite.

\noindent Similarly,\\
$\left\{\int \frac{(\mu_{\bX\bX'}\sum_{k<l}s_{k,l}+1)^{n/2}}{|\bX\bD\bX'+\bI|^{1/2}(\iota+\sum\limits_{k<l}s_{k,l})^{q+\zeta}}
d\prod_{k<l}s_{k,l}\right\}\leq \left\{\int \frac{(\mu_{\bX\bX'}\sum_{k<l}s_{k,l}+1)^{n/2}}{(\mu_{\bX\bX',min}\min_{k<l}s_{k,l}+1)^{n/2}(\iota+\sum\limits_{k<l}s_{k,l})^{q+\zeta}}
d\prod_{k<l}s_{k,l}\right\},$ where $\mu_{\bX\bX',min}$ is the minimum eigenvalue of $\bX\bX'$. The last inequality follows from the fact that
$\bX\bX'\geq \mu_{\bX\bX',min}\min_{k<l}s_{k,l}\bI$.  $\left\{\int \frac{(\mu_{\bX\bX'}\sum_{k<l}s_{k,l}+1)^{n/2}}{(\mu_{\bX\bX',min}\min_{k<l}s_{k,l}+1)^{n/2}(\iota+\sum\limits_{k<l}s_{k,l})^{q+\zeta}}
d\prod_{k<l}s_{k,l}\right\}$ is finite if and only if $\left\{\int \frac{(\mu_{\bX\bX'}\sum_{k<l}s_{k,l}+1)^{n/2}}{(\mu_{\bX\bX',min}\sum\limits_{k<l}s_{k,l}+1)^{n/2}(\iota+\sum\limits_{k<l}s_{k,l})^{q+\zeta}}
d\prod_{k<l}s_{k,l}\right\}<\infty$, by ratio test of integrals. Since the latter integral is finite, $\int\limits_{\mathcal{A}}\tilde{p}(\bu_1,..,\bu_V,\tau^2,\{s_{k,l}\}_{k<l})d\bu_1\cdots d\bu_V d\tau^2 d\prod\limits_{k<l}s_{k,l}\leq\infty$.

Now consider the expression $\int\limits_{\mathcal{A}^c}\tilde{p}(\bu_1,..,\bu_V,\tau^2,\{s_{k,l}\}_{k<l})d\tau^2 d\prod\limits_{k<l}s_{k,l}d\bu_1\cdots d\bu_V$. It is easy to see that $\mathcal{A}^c=\{(\bu_1,...,\bu_V):||\by-\bX\bW||^2\leq 1\}$ is a bounded set, so that the bounded function \\
$\exp\left\{-\frac{||\by-\bX\bW||^2}{2\tau^2(\mu_{\bX\bX'}\sum_{k<l}s_{k,l}+1)}\right\}$ achieves the maximum value at $\bW=\bW^*$. Thus,
\begin{align*}
& \int\limits_{\mathcal{A}^c}\tilde{p}(\bu_1,..,\bu_V,\tau^2,\{s_{k,l}\}_{k<l})d\bu_1\cdots d\bu_V d\tau^2 d\prod\limits_{k<l}s_{k,l}
\leq \mbox{constant}\int\limits_{\mathcal{A}^c}\frac{1}{|\bS+\sum_{k=1}^{V}\bu_k\bu_k'|^{(\nu+V)/2}}\times\\
& \qquad \exp\left\{-\frac{||y-X\bW^*||^2}{2\tau^2(\mu_{\bX\bX'}\sum_{k<l}s_{k,l}+1)}\right\}\times
\int \frac{1}{(\iota+\sum\limits_{k<l}s_{k,l})^{q+\zeta}}\frac{d\tau^2 d\prod\limits_{k<l}s_{k,l}}{(\tau^2)^{n/2+1}}\frac{1}{|\bX\bD\bX'+\bI|^{1/2}}\\
&\leq \frac{\mbox{constant}}{||y-X\bW^*||^n}\left\{\int\limits_{\mathcal{A}^c} \frac{1}{|\bS+\sum_{k=1}^{V}\bu_k\bu_k'|^{(\nu+V)/2}}d\bu_1\cdots d\bu_V\right\}\times\\
&\qquad
\left\{\int \frac{(\mu_{\bX\bX'}\sum_{k<l}s_{k,l}+1)^{n/2}}{|\bX\bD\bX'+\bI|^{1/2}(\iota+\sum\limits_{k<l}s_{k,l})^{q+\zeta}} d\prod_{k<l}s_{k,l}\right\}<\infty,
\end{align*}
where the last step follows from earlier discussions.

\bibliographystyle{natbib}
\bibliography{sample1}
\end{document}